\documentclass[conference]{IEEEtran}
\IEEEoverridecommandlockouts
\usepackage{amsmath,amssymb,amsfonts,amsthm}
\usepackage{nccmath}
\usepackage{algorithmic}

\usepackage{cite}

\usepackage{graphicx}
\usepackage{xcolor}

\usepackage{textcomp}

\usepackage{newtxtext}
\usepackage[varg]{newtxmath}
\usepackage{comment}

\usepackage{hyperref}

\newcommand{\Prb}{\mathsf{P}}\newcommand{\Exp}{\mathsf{E}}

\newcommand{\dd}{\mathrm{d}}\newcommand{\ee}{\mathrm{e}}
\newcommand{\R}{\mathbb{R}}
\newcommand{\N}{\mathbb{N}}

\newcommand{\bm}[1]{\boldsymbol{#1}}

\newcommand{\argmax}{\mathop{\rm arg~max}\limits}
\newcommand{\argmin}{\mathop{\rm arg~min}\limits}

\newtheorem{proposition}{Proposition}
\newtheorem{theorem}{Theorem}
\newtheorem{corollary}{Corollary}
\newtheorem{lemma}{Lemma}
\newtheorem{remark}{Remark}

\newtheorem{approximation}{Approximation}

\def\BibTeX{{\rm B\kern-.05em{\sc i\kern-.025em b}\kern-.08em
    T\kern-.1667em\lower.7ex\hbox{E}\kern-.125emX}}
\begin{document}

\title{Dynamic Modeling of Target Cell Location for Mobility Robustness Analysis in Cellular Networks: Technical Report
%\thanks{
%This work was supported in part by a research grant from the Okasan-Kato Foundation, Tsu, Japan.
%The author would like to thank Prof. Naoto Miyoshi for valuable discussions.
%}
}

%\and

\author{
    \IEEEauthorblockN{Kiichi Tokuyama}
    \IEEEauthorblockA{
    Graduate School of Engineering, Mie University, Japan\\
    tokuyama@info.mie-u.ac.jp
    }
}

\maketitle

\begin{abstract}
Mobility robustness optimization (MRO) requires an appropriate selection
of handover (HO) parameters such as the time-to-trigger (TTT) and offset
margin to balance HO failures and ping-pong HOs.
Existing stochastic geometry-based analyses for MRO have treated 
the angular position of the target base station (BS) as uniformly distributed 
over a feasible region.
However, this treatment does not explicitly capture the spatial
distribution of the target BS dynamically selected as a user equipment
(UE) moves through the network.
In this paper, we develop a stochastic geometry-based analytical
framework for MRO in sub-6 GHz cellular networks.
We derive the distribution of the HO triggering time and the spatial
distribution of the dynamically selected target BS under straight-line
UE mobility.
Based on these distributions, we formulate too-late HO and ping-pong HO
events as mutually exclusive events and analytically derive
their probabilities.
Numerical results validate the analysis, demonstrate improved accuracy
over the conventional uniform-angle model, and reveal the tradeoff
between the two HO events and the dependence of the optimal TTT on BS
density.
\end{abstract}

%\begin{abstract}
%This document is a model and instructions for \LaTeX.
%This and the IEEEtran.cls file define the components of your paper [title, text, heads, etc.]. *CRITICAL: Do Not Use Symbols, Special Characters, Footnotes, 
%or Math in Paper Title or Abstract.
%\end{abstract}

\begin{IEEEkeywords}
Mobility robustness optimization, target cell, handover failure, ping-pong handover, stochastic geometry 
\end{IEEEkeywords}

\section{Introduction}

%Mobile cellular communications have continuously evolved from 5G New Radio (NR) toward 6G, while terrestrial cellular networks remain an important communication infrastructure.
%Ensuring seamless and stable communications for mobile user equipments (UEs) therefore remains an important research challenge.
%
%Mobility robustness optimization (MRO), a functionality of self-organizing networks, has been extensively studied to automatically adjust handover (HO) control parameters such as Time-to-Trigger (TTT), Hysteresis, and A3 offset \cite{nguyen2018mobility, tashan2022mobility}.
%Smaller HO parameter values generally reduce handover failures (HOFs) at the cost of increasing ping-pong (PP) HOs, whereas larger values suppress PP HOs while increasing HOFs.
%Appropriate parameter optimization is therefore essential to balance these conflicting HO events.

Ensuring seamless and stable communications for mobile user equipment (UEs) remains an important challenge in terrestrial cellular networks, for which mobility robustness optimization (MRO) plays a key role. As a functionality of self-organizing networks, MRO has been extensively studied to automatically adjust handover (HO) control parameters such as Time-to-Trigger (TTT), Hysteresis, and A3 offset \cite{nguyen2018mobility, tashan2022mobility}. Smaller HO parameter values generally reduce handover failures (HOFs) at the cost of increasing ping-pong (PP) HOs, whereas larger values suppress PP HOs while increasing HOFs. Appropriate parameter optimization is therefore essential to balance these conflicting HO events.

Various MRO approaches have been investigated to optimize HO parameters while addressing the tradeoff between HOFs and PP HOs. These include machine learning-based optimization methods \cite{nguyen2021machine,Karmakar2023mobility}, theoretical approaches based on Markov chains \cite{guidolin2016context} and Stochastic Petri Nets \cite{zheng2025directional}, and deterministically geometric approaches that model cell boundaries and UE mobility to analyze HOF and PP rates \cite{perez2012theoretical,vasudeva2017analysis,nguyen2020geometry}. 
%In particular, stochastic geometry based on Poisson point processes (PPPs) has been widely employed to evaluate HO performance in spatially random cellular networks \cite{...}. By capturing the spatial distribution of BSs and the mobility characteristics of UEs, stochastic geometry provides a powerful analytical framework for HO performance evaluation and offers useful design insights for MRO.

Stochastic geometry has also been widely employed for HO performance analysis in spatially random cellular networks \cite{lima2014modeling,salehi2021stochastic,xu2017modeling,zhou2023heterogeneous,guo20233d,wei2023equivalent,wei2024time}.
Earlier studies investigated HOFs and PP HOs under various network configurations. 
In \cite{lima2014modeling}, the HOF rate was analytically derived for UEs moving along straight trajectories in small-cell networks, while \cite{salehi2021stochastic} evaluated the PP rate by using the serving-cell sojourn-time distribution in multi-tier heterogeneous networks (HetNets). 
These studies considered either the HOF rate or the PP rate.
In \cite{xu2017modeling}, both HOF and PP rates were analytically derived
for two-tier HetNets. This analytical framework was
subsequently extended to three-tier HetNets with traffic
hotspots \cite{zhou2023heterogeneous} and two-tier HetNets with cellular-connected unmanned aerial vehicles (UAVs) \cite{guo20233d}. 
However, the studies in \cite{xu2017modeling, zhou2023heterogeneous, guo20233d} characterized HOFs mainly based on geometric relationships between the serving and target BSs, without explicitly accounting for aggregate interference from surrounding BSs.

More recently, frameworks for HOF analysis incorporating aggregate interference have been developed. In \cite{wei2023equivalent}, both HOF and PP rates were analytically derived for multi-tier UAV networks by transforming UAV altitudes into equivalent distances on the ground plane, thereby establishing an analytical framework equivalent to that for terrestrial cellular networks. This framework was further extended to UAV-assisted multi-tier HetNets with channel fading \cite{wei2024time}. 
However, in \cite{wei2023equivalent,wei2024time}, the angular position of the target BS, conditioned on its distance from the HO triggering location, is treated as uniformly distributed over the feasible region.
This treatment does not explicitly capture the target BS distribution induced by UE mobility.
To the best of our knowledge, no analytical framework has characterized this mobility-induced target BS distribution and evaluated its impact on HO performance for MRO.

%Various MRO approaches have been investigated, including machine learning-based optimization \cite{nguyen2021machine,Karmakar2023mobility} and analytical approaches based on Markov chains \cite{guidolin2016context}, Stochastic Petri Nets \cite{zheng2025directional}, and deterministic geometric models \cite{perez2012theoretical,vasudeva2017analysis,nguyen2020geometry}.
%Stochastic geometry has also been widely employed for HO performance analysis in spatially random cellular networks \cite{lima2014modeling,salehi2021stochastic,xu2017modeling,zhou2023heterogeneous,guo20233d,wei2023equivalent,wei2024time}.
%
%While earlier studies analyzed HOFs and/or PP HOs mainly based on geometric relationships between serving and target BSs \cite{lima2014modeling,salehi2021stochastic,xu2017modeling,zhou2023heterogeneous,guo20233d}, recent studies \cite{wei2023equivalent,wei2024time} incorporated aggregate interference from surrounding BSs into the HOF analysis. In particular, these studies analytically derived both HOF and PP rates for multi-tier unmanned aerial vehicle (UAV) networks by transforming UAV altitudes into equivalent distances on the ground plane, thereby establishing an analytical framework analogous to that for terrestrial cellular networks.

In this paper, we develop a stochastic geometry-based analytical framework for HO performance evaluation in Sub-6 GHz cellular networks.
The main contributions are summarized as follows:
\begin{itemize}
    \item
    We analytically derive the mobility-induced target BS distribution, revealing its generally non-uniform angular distribution, and use it as the basis for subsequent HO performance analysis.
    \item
    We formulate too-late HOs, a type of HOF \cite{nguyen2018mobility, nguyen2020geometry}, and PP HOs as mutually exclusive events and analytically derive their probabilities while accounting for aggregate interference from surrounding BSs.
    \item
    Numerical results demonstrate the non-negligible impact of the derived target BS distribution through comparison with a baseline model. They also characterize the tradeoff between the too-late HO and PP rates with respect to TTT and identify the TTT value that minimizes their sum.
\end{itemize}

\section{System Model}

\subsection{Sub-6 GHz Cellular Networks}

We consider a single-tier cellular network where BS locations follow a homogeneous Poisson point process (PPP) $\Phi=\{X_1,X_2,\ldots\}$ on $\R^2$ with intensity $\lambda>0$.
Each BS $X_k\in\R^2$ has its corresponding cell $k\in\N$.
All BSs operate on the same carrier frequency with omnidirectional antennas and constant transmit power.
The omnidirectional assumption has been widely adopted in the literature for sub-6 GHz communications, e.g., \cite{shi2019coverage,muhammad2021stochastic}.
The downlink channels follow power-law path loss with exponent $\beta>2$ and i.i.d. Rayleigh block fading \cite{haenggi2012stochastic}, while shadowing is ignored.

The received signal power at a UE located at $\bm{u}_t\in\R^2$ at time $t\geq0$ from cell $k$ is
\begin{equation*}%\label{eq:received_signal_power}
  H_{k,n(t)}\|X_k-\bm{u}_t\|^{-\beta},
\end{equation*}
where $n(t)$ denotes the coherence-interval index and $\{H_{k,n(t)}\}$ are i.i.d. unit-mean exponential random variables.
The transmit power is normalized to unity.
Assuming that every BS serves at least one UE, the downlink signal-to-interference-plus-noise ratio (SINR) from serving cell $i$ is
\begin{equation}\label{eq:SINR}
  \mathsf{SINR}_{i}(t)
  =
  \frac{H_{i,n(t)}\|X_i-\bm{u}_t\|^{-\beta}}
  {\sum_{j\in\N\backslash\{i\}}H_{j,n(t)}
  \|X_j-\bm{u}_t\|^{-\beta}+\sigma^2},
\end{equation}
where $\sigma^2\geq0$ is the noise power.

\subsection{RRC Handover Procedures}

We model the radio resource control (RRC) procedures relevant to HO execution and radio link failure (RLF) detection, which are essentially common to 4G LTE and 5G NR.

\subsubsection{Triggering Condition of Event A3}

We focus on Event A3, which compares the measured signal strengths of the serving and neighboring cells.
Assuming perfect measurement filtering\footnote{Perfect Layer-1 and Layer-3 filtering is commonly assumed
in the literature, e.g., \cite{xu2017modeling,guo20233d}.}, the reference signal received power (RSRP) from cell $k$, expressed in dBm, is
\begin{equation}\label{eq:RSRP_def}
  \mathsf{RSRP}_k(t)
  =10\log_{10}\|X_k-\bm{u}_t\|^{-\beta}.
\end{equation}
According to \cite{3gppTS36331}, for serving cell $i$ and neighboring cell $j$, the triggering condition of Event A3 (hereafter, the A3 condition) can be written as
\begin{equation}\label{eq:condition_eventA3}
  \mathsf{RSRP}_j(t)+\mathrm{CIO}_j-\mathrm{Hys}
  \geq
  \mathsf{RSRP}_i(t)+\mathrm{CIO}_i+\mathrm{Off_{A3}},
\end{equation}
where $\mathrm{Hys}$ and $\mathrm{Off_{A3}}$ take values in $[0,15]~\mathrm{dB}$ and $[-15,15]~\mathrm{dB}$, respectively.
The frequency-specific offsets are removed in \eqref{eq:condition_eventA3} since we assume the same carrier frequency.

Assuming identical $\mathrm{CIO}_k$ for all cells and defining
\begin{equation*}%\label{eq:effective_offset}
  \mathsf{Off}=\mathrm{Hys}+\mathrm{Off_{A3}},
\end{equation*}
the A3 condition becomes
\begin{align}\label{eq:trigger_cond_equiv}
  \hat{a}\|X_i-\bm{u}_t\|\geq\|X_j-\bm{u}_t\|,
\end{align}
where we assume $\mathsf{Off}<0$ for analytical tractability and
\begin{equation}\label{eq:a_hat}
  \hat{a}=10^{-\frac{\mathsf{Off}}{10\beta}}>1.
\end{equation}

\subsubsection{RRC Connection Re-establishment}

We focus on a UE initially undergoing RRC connection re-establishment and assume that it immediately associates with the cell providing the largest RSRP.
The initial serving cell $I$ is thus
\begin{align}\label{eq:initial_serving_cell_t0}
  I=\argmax_{k\in\N}\mathsf{RSRP}_k(0).
\end{align}
We define the initial non-triggering event as
\begin{align}
  \mathcal{N}_0
  =
  \left\{
    \hat{a}\|X_I-\bm{u}_0\|<\|X_j-\bm{u}_0\|,
    \ \forall j\in\N\backslash\{I\}
  \right\},
  \label{eq:initial_non-triggering_event}
\end{align}
under which the A3 condition does not hold at $t=0$ for any other cells.

\subsubsection{TTT-based HO Execution}

Let $\Delta_{\mathsf{T}}>0$ denote the common TTT duration.
For the initial serving cell $I$ and neighboring cell $j$, define the first A3 triggering time $\tilde{T}_{I,j}$ and the triggering time $T_{I,j}$ leading to TTT expiration as
\begin{align}
  \tilde{T}_{I,j}
  &=
  \inf\bigl\{\tau\geq0:
  \hat{a}\|X_I-\bm{u}_\tau\|\geq\|X_j-\bm{u}_\tau\|\bigr\},
  \label{eq:TTT_trigger_Sj}
\\
  T_{I,j}
  &=
  \inf\bigl\{\tau\geq0:
  \hat{a}\|X_I-\bm{u}_s\|\geq\|X_j-\bm{u}_s\|
  \ \forall s\in[\tau,\tau+\Delta_{\mathsf{T}})\bigr\},
  \label{eq:TTT_expire_ij}
\end{align}
where they are set to infinity if the corresponding sets are empty.
The first neighboring cell and the target cell are defined as
\begin{align}
  \mathsf{f}(I)
  &\in\argmin_{j\in\N\backslash\{I\}}\tilde{T}_{I,j},
  \label{eq:ho_time_cell_ini}
\\
  \mathsf{t}(I)
  &\in\argmin_{j\in\N\backslash\{I\}}T_{I,j}.
  \label{eq:ho_time_cell_fin}
\end{align}
The HO completion time from $I$ to $\mathsf{t}(I)$ is given by $T_{I,\mathsf{t}(I)} + \Delta_{\mathsf{T}}$, and
the target BS location is denoted by $X_{\mathsf{t}(I)}$.

\subsubsection{Out-of-sync Indication and RLF Detection}

We abstract the N310/T310-based RLF mechanism \cite{3gppTS36331} by an RLF timer of duration $\Delta_{\mathsf{F}}>0$.
An out-of-sync indication is generated when the radio link quality falls below the threshold $Q_{\mathsf{F}}$, corresponding to $Q_{\mathrm{out}}$ in 3GPP TS 36.133 \cite{3gppTS36133}.
Conditioned on $X_I$ and $X_{\mathsf{t}(I)}$, we define the radio link quality\footnote{A similar use of the conditional expectation of the SINR for
characterizing the radio link quality appears in the analysis of \cite{wei2023equivalent}.} as
\begin{equation}\label{eq:downlink_RLQ}
  \rho(t\mid X_I,X_{\mathsf{t}(I)})
  =
  \Exp[\mathsf{SINR}_I(t)
  \mid\mathcal{N}_0,X_I,X_{\mathsf{t}(I)}].
\end{equation}
The RLF detection time is then
\begin{align}
\label{eq:RLF_expire_i}
  F_I
  =
  \inf\bigl\{
    \tau\geq\Delta_{\mathsf{F}}:
    \rho(s\mid X_I,X_{\mathsf{t}(I)})<Q_{\mathsf{F}},
    \ \forall s\in[\tau-\Delta_{\mathsf{F}},\tau)
  \bigr\},
\end{align}
where $F_I=\infty$ if the set is empty.
The initial in-sync event is
\begin{equation}
  \mathcal{S}_0
  =
  \{\rho(0\mid X_I,X_{\mathsf{t}(I)})\geq Q_{\mathsf{F}}\},
  \label{eq:initial_in-sync_event}
\end{equation}
under which no out-of-sync indication is generated at time $0$.

\subsection{UE Mobility}

We assume that each UE moves in a straight line with a constant velocity on $\R^2$.
By the spatial stationarity and isotropy of the network model, we place the typical UE that performs RRC connection re-establishment at $t = 0$ at the origin $\bm{o} = (0, 0) \in \R^2$ and assume that it moves along the horizontal axis without loss of generality.
%We refer to this UE as the typical UE.
%
Then, we have $\bm{u}_t = (tv, 0)$, where $v > 0$ denotes the speed of the typical UE.

%By stationarity and isotropy, we place the typical UE at the origin at $t=0$ and assume straight-line motion along the horizontal axis with constant speed $v>0$, i.e.,
%$\bm{u}_t=(tv,0)$.

\subsection{HO Events with Initial Conditions}\label{subsec:ho_event}

Following \cite{nguyen2020geometry}, we consider too-late and ping-pong HO events.

\subsubsection{Too-late HO Event}

A too-late HO occurs when RLF detection precedes HO completion, represented by
\begin{equation}\label{eq:ev_too_late}
  \mathcal{L}
  =
  \mathcal{N}_0\cap\mathcal{S}_0
  \cap
  \left\{
  T_{I,\mathsf{t}(I)}+\Delta_{\mathsf{T}}\geq F_I
  \right\}.
\end{equation}

\subsubsection{Ping-pong HO Event}

A ping-pong HO occurs when, immediately after the HO from $I$ to $\mathsf{t}(I)$, the reverse A3 condition toward $I$ remains satisfied for $\Delta_{\mathsf{T}}$.
The corresponding event is
\begin{align}\label{eq:ev_ping_pong}
  \mathcal{P}
  =
  \mathcal{N}_0\cap\mathcal{S}_0
  \cap
  \left\{
  T_{\mathsf{t}(I),I}
  =
  T_{I,\mathsf{t}(I)}+\Delta_{\mathsf{T}}<F_I
  \right\},
\end{align}
where
\begin{align}\label{eq:TTT_expire_tI_to_I}
  T_{\mathsf{t}(I),I}
  =
  \inf\bigl\{\tau\geq T_{I,\mathsf{t}(I)}+\Delta_{\mathsf{T}}:\,
  &\hat{a}\|X_{\mathsf{t}(I)}-\bm{u}_s\|
  \geq\|X_I-\bm{u}_s\|
  \nonumber\\
  &\forall s\in[\tau,\tau+\Delta_{\mathsf{T}})\bigr\}.
\end{align}

\begin{remark}
For $\Delta_{\mathsf{T}}<\Delta_{\mathsf{F}}$, the events
$\mathcal{L}$ and $\mathcal{P}$ reduce to the too-late and PP HO events in
\cite{nguyen2020geometry}, respectively, except that the initial conditions
$\mathcal{N}_0$ and $\mathcal{S}_0$ are additionally imposed.
\end{remark}

\section{Distribution of the Cells}

%In this section, we analyze the statistical characteristics of the cells surrounding the typical UE needed to analyze the events $\mathcal{L}$ and $\mathcal{P}$ defined in the preceding section.
In this section, we investigate the statistical distributions of the cells surrounding the typical UE.
These results are used to analyze the events $\mathcal{L}$ and $\mathcal{P}$ defined in the preceding section.
%In this section, we investigate the distributions of the cells surrounding the typical UE, with a view to analyzing the events $\mathcal{L}$ and $\mathcal{P}$ defined in the preceding section.
%
%Here, the location of cell $i \in \N$ refers to the location of $X_i \in \Phi$.
%
The goal of this section is to characterize $X_{\mathsf{t}(I)} \in \Phi$, i.e., the BS location of the target cell $\mathsf{t}(I)$, under the initial non-triggering event $\mathcal{N}_0$.
% via the triggering time $T_{I, \mathsf{t}(I)}$ in \eqref{eq:TTT_expire_ij}. 
%Here, the initial serving cell $I$ is given in \eqref{eq:initial_serving_cell_t0}.
As part of the analysis, the distribution of the triggering time $T_{I, \mathsf{t}(I)}$ is also derived.
Hereafter, we adopt polar coordinates to represent $\bm{x} \in \R^2$ and write $\bm{x} = (x, \phi)$, where $x \ge 0$ and $|\phi| \le \pi$.

%\textcolor{white}{.}\\（第2段落）
\subsection{Initial Serving Cell and Non-triggering Event}

We discuss the distribution of $X_I \in \Phi$ and the probability of $\mathcal{N}_0$.
By the definition $I$ given in \eqref{eq:initial_serving_cell_t0}, together with \eqref{eq:RSRP_def} and $\bm{u}_0 = \bm{o}$, the joint probability density function (PDF) of $X_I$, denoted by $f_{X_I}(r, \theta)$ for $r \ge 0$ and $|\theta| \le \pi$, is well known as \cite{chiu2013stochastic}
\begin{equation}\label{eq:joint_pdf_XI}
  f_{X_I}(r, \theta) = \lambda r \ee^{-\pi\lambda r^2}.
\end{equation}
Given $X_I = (r, \theta)$, we regard $r > 0$ and $|\theta| < \pi$.
By the definition of $\mathcal{N}_0$ in \eqref{eq:initial_non-triggering_event}, $\mathcal{N}_0$ occurs if and only if no points of $\Phi$ lie on the region $b_{\bm{o}}(\hat{a} \|X_I\|) \backslash b_{\bm{o}}(\|X_I\|)$. 
Therefore, 
\begin{equation}\label{eq:cond_prob_N0}
  \Prb\bigl(\mathcal{N}_0 | X_I = (r, \theta)\bigr)
  =
  \ee^{-\pi\lambda r^2(\hat{a}^2 - 1)}.
\end{equation}
%

%\subsubsection{The first triggered time distribution}
\subsection{Triggering Time Distribution}

Recall that $T_{I,\mathsf{t}(I)}$
%, defined in \eqref{eq:TTT_expire_ij} with \eqref{eq:ho_time_cell_fin}, 
denotes the triggering time at which the TTT timer expires after $\Delta_{\mathsf{T}}$, for the initial serving cell $I$ and its target cell $\mathsf{t}(I)$ given in \eqref{eq:ho_time_cell_fin}.
%The following holds for $T_{I,\mathsf{t}(I)}$.

\begin{lemma}\label{lem:T_tildeT_equiv}
Under the initial non-triggering event $\mathcal{N}_0$, 
%it holds that
%
\begin{equation}\label{eq:T_tildeT_equiv}
  T_{I,\mathsf{t}(I)} = \tilde{T}_{I, \mathsf{f}(I)} > 0,
\end{equation}
where $\tilde{T}_{I, \mathsf{f}(I)}$ is the first triggering time, with $\mathsf{f}(I)$ given in \eqref{eq:ho_time_cell_ini}.
\end{lemma}
\begin{proof}

For $j \in \N\backslash\{I\}$, define $D_{I, j}$ as
\begin{equation*}%\label{eq:non_trigger_domain}
    D_{I, j} = \{\bm{z} \in \R^2: \hat{a} \|X_I - \bm{z}\| < \|X_j - \bm{z}\| \},
\end{equation*}
where $\hat{a} > 1$ is given in \eqref{eq:a_hat}. 
$D_{I, j}$ represents the region in $\R^2$ where the A3 condition in \eqref{eq:trigger_cond_equiv} is not satisfied for $I$ and $j$.
Under the event in \eqref{eq:initial_non-triggering_event}, $\bm{u}_0 = \bm{o} \in D_{I, j}$, which is an open disk since $\hat{a} > 1$.
Thus, $\tilde{T}_{I, j} > 0$ follows from \eqref{eq:TTT_trigger_Sj}.
%
%
%
%Since $\bm{u}_t = (tv, 0)$, i.e., the UE moves on a straight line, $\bm{u}_t \notin D_{I, j}$ for any $t \ge t'$ when $t' \notin D_{I, j}$,
%
%which implies that the A3 condition is retained ever since it is once satisfied.
%
%
Since $\bm{u}_t = (tv, 0)$, if $\bm{u}_{t'} \notin D_{I, j}$ for some $t'$, then $\bm{u}_t \notin D_{I, j}$ for all $t \ge t'$,
which implies that the A3 condition, once satisfied, remains satisfied thereafter.
Therefore, from \eqref{eq:TTT_trigger_Sj} and \eqref{eq:TTT_expire_ij}, we have
$\tilde{T}_{I, j} = T_{I, j}$.
The claim follows from \eqref{eq:ho_time_cell_ini} and \eqref{eq:ho_time_cell_fin}.
\end{proof}

%The result of the distribution of $T_{I, \mathsf{t}(I)}$ is as follows.
%The distribution of $T_{I, \mathsf{t}(I)}$ is characterized as follows.

\begin{proposition}\label{prop:pdf_TS0}
Under the initial non-triggering event $\mathcal{N}_0$ and given $X_I = (r, \theta)$, the conditional PDF of $T_{I, \mathsf{t}(I)}$ for $t > 0$ is
\begin{align}\label{eq:pdf_TS0}
  &f_{T_{I, \mathsf{t}(I)} | \mathcal{N}_0}(t \,|\, r, \theta)
  \\ 
  &= 
  \begin{cases}
    \begin{aligned}
      & \lambda \frac{\dd}{\dd t} \xi(t, r, \theta) \exp\bigl( -\lambda \xi(t, r, \theta)\bigr),
      \\[6pt]
      &\quad \text{when } \ \textstyle 
%        \frac{\hat{K}_0 r}{v} (\hat{a}\cos\theta - 1)^+ \le t \le\frac{\hat{K}_0 r}{v} (\hat{a}\cos\theta + 1)^+, 
         |\hat{a}r - tv| \le \hat{r}(t) \le \hat{a}r + tv,
      \\[6pt]
    \end{aligned}
    \\[6pt]
    \begin{aligned}
      & 2\pi\lambda\hat{a}^2 v (tv - r\cos\theta) \exp\bigl( -\pi\lambda\hat{a}^2tv(tv - 2r\cos\theta) \bigr),
      \\[6pt]
      &\quad \text{when } \ \textstyle 
%        \frac{\hat{K}_0 r}{v} (\hat{a}\cos\theta + 1)^+ \le t,
         \hat{a}r + tv \le \hat{r}(t),
      \\[6pt]
    \end{aligned}
    \\[6pt]
    0,
      \quad \text{otherwise}, 
  \end{cases}
\nonumber
\end{align}
where 
%$w_{r, tv, \theta} = \sqrt{r^2 + t^2 v^2 - 2rtv\cos\theta}$, and
%
\begin{align*}
  \textstyle
  &\xi(t, r, \theta)
  =
  \hat{a}^2 w_{r, tv, \theta}^2
  \cos^{-1}
  \left(
    \frac{\hat{a}r\cos\theta}{w_{r \!, tv \!, \theta}}
    -
%    \frac
%      {\hat{a}^2 + 1}
%      {2\hat{a}}
    \frac{tv}{w_{r \!, tv \!, \theta}}
    A^{\!+}
  \right)
  \\
  &-
  \hat{a}^2 r^2
  \cos^{-1} \!
  \left(
    \hat{a}\cos\theta
    -
%    \frac
%      {\hat{a}^2 - 1}
%      {2\hat{a}}
    \frac{tv}{r} A^{\!-}
  \right)
  +
  \hat{a}rtv
  \sqrt{
    1
    -
    \left(
      \hat{a}\cos\theta
      -
%      \frac
%        {\hat{a}^2 - 1}
%        {2\hat{a}}
      \frac{tv}{r} A^{\!-}
    \right)^2
  },
\end{align*}
with $A^{\!+} = (\hat{a}^2 + 1)/(2\hat{a})$, and $A^{\!-} = (\hat{a}^2 - 1)/(2\hat{a})$, and
\begin{equation}\label{eq:r_hat_t}
    \hat{r}(t) = \hat{a} w_{r, tv, \theta} = \hat{a} \sqrt{r^2 + t^2 v^2 - 2rtv\cos\theta}.
\end{equation}
\end{proposition}

\begin{remark}%\label{rmk:cc_intersect_comp}
Although $\xi(t, r, \theta)$ admits a closed-form expression,
its derivative with respect to $t$ becomes algebraically involved.
For numerical evaluation of the PDF, we therefore compute 
$\frac{\mathrm d}{\mathrm dt}\xi(t, r, \theta)$ numerically.
The expression in \eqref{eq:pdf_TS0} does not apply to the case $\mathsf{Off} > 0$, in which Lemma~\ref{lem:ball_inclusion} given below does not hold.
\end{remark}

\begin{lemma}\label{lem:ball_inclusion}
Let $\ell>0$, $a \ge 1$, and $\bm{b} \in \R^2$. 
Define the disk region $d(u)$ for $\bm{u}=(u,0)$, $0\le u\le\ell$, as
\[
d(u)=b_{\bm{u}}\bigl( a\|\bm{b}-\bm{u}\| \bigr),
\]
where $b_{\bm{x}}(r)$ denotes an open ball in $\R^2$ with center $\bm{x}$ and radius $r>0$. Then
\begin{equation}\label{eq:ball_inclusion}
  \bigcup_{0\le u\le\ell} d(u) = d(0)\cup d(\ell).
\end{equation}
\end{lemma}
\begin{proof}
For fixed $\bm{x} \in \R^2$, define the function 
\(
g(u)=a^2\|\bm{b}-\bm{u}\|^2-\|\bm{x}-\bm{u}\|^2.
\)
Note that $\bm{x} \in d(u)$ if and only if $g(u) > 0$.
Since $a \ge 1$, $g$ is a linear function or a convex quadratic function of $u$.
Hence, $g(u)$ attains its maximum over the compact interval $[0,\ell]$
when $u\in\{0,\ell\}$.
%Hence, by the standard fact that a convex function on a compact 
%interval attains its maximum at an extreme point of the interval
%(e.g., Bauer's maximum principle~\cite{ConvexOpt}),
%$\max_{u\in[0,\ell]} g(u)$ is attained at $u\in\{0,\ell\}$.
This implies that $\bm{x}\in d(u)$ for some $u\in[0,\ell]$
if and only if $\bm{x}\in d(0)\cup d(\ell)$.
\end{proof}

\begin{proof}[Proof of Proposition~\ref{prop:pdf_TS0}]
By \eqref{eq:TTT_trigger_Sj} and \eqref{eq:ho_time_cell_ini}, 
$\tilde{T}_{I, \mathsf{f}(I)} \le t$ holds if and only if
\begin{equation}\label{eq:non-trigger_cond_equiv}
  \exists\, (t', j) \in [0,t] \times (\N\setminus\{I\})
  \ \text{ s.t. }\
  \hat{a} \|X_I - \bm{u}_{t'}\| \ge \|X_j - \bm{u}_{t'}\|.
\end{equation}
Let $|B|$ denote the area of $B \subset \R^2$ and define the conditional non-empty probability of the homogeneous PPP $\Phi$ on $B$ as
\begin{equation}\label{eq:comp_void_prob_PPP}
  \bar{\mathcal{V}}(B \mid r, \theta) := \Prb\bigl(\Phi(B) \neq 0 \mid \mathcal{N}_0, X_I = (r, \theta)\bigr).
\end{equation}
Then, we have
\begin{align}\label{eq:ccdf_TS0}
  &\Prb\bigl(T_{I, \mathsf{t}(I)} \le t \mid \mathcal{N}_0, X_I = (r, \theta)\bigr)
  \nonumber\\
  &=
  \bar{\mathcal{V}}\Bigl( 
    \bigcup_{t' \in [0, t]} b_{\bm{u}_{t'}} \bigl(\hat{a}\|X_I - \bm{u}_{t'}\| \bigr) \,\Big|\,\, r, \theta
  \Bigr)
  \nonumber\\
  &=
  \bar{\mathcal{V}}\Bigl( 
     b_{\bm{o}} \bigl(\hat{a}\|X_I\| \bigr) 
     \cup 
     b_{\bm{u}_{t}} \bigl(\hat{a}\|X_I - \bm{u}_{t}\| \bigr)
     \,\Big|\,\, r, \theta
  \Bigr)
  \nonumber\\
  &=
  1 - \exp\Bigl( -\lambda \,\big| b_{(tv, 0)}(\hat{a} w_{r, tv, \theta}) \backslash b_{\bm{o}}(\hat{a}r) \big|\, \Bigr),
\end{align}
where 
the first equality follows from \eqref{eq:T_tildeT_equiv} and \eqref{eq:non-trigger_cond_equiv},
\eqref{eq:ball_inclusion} is applied in the second equality, and
the last equality follows from the void probability of $\Phi$ conditioned on $X_I = (r, \theta)$ and $\mathcal{N}_0$ in \eqref{eq:initial_non-triggering_event}.
Here, $w_{r, tv, \theta} = \sqrt{r^2 + t^2 v^2 - 2rtv\cos\theta}$.
By the well-known expression for the area of intersection between two circles,
%
%
%\begin{fleqn}[10pt]
\begin{align}\label{eq:cc_intersect_comp}
  &|b_{(\ell, 0)}(q) \backslash b_{\bm{o}}(p)|
  \\
  &=
  \begin{cases}
    \begin{aligned}
      & q^2\cos^{-1}\!\left(\frac{p^2 \!-\! q^2 \!-\! \ell^2}{2q\ell}\right)
        - p^2\cos^{-1}\!\left(\frac{p^2 \!-\! q^2 \!+\! \ell^2}{2p\ell}\right) \\
      &\,
        + p\ell\sqrt{1 - \left(\frac{p^2 \!-\! q^2 \!+\! \ell^2}{2p\ell}\right)^{\! 2}},
        \quad \!\! \text{when } |p-\ell| \le q \le p + \ell,      
    \end{aligned}
    \\[6pt]
    \pi (q^2 - p^2),
      \quad \! \text{when } p+\ell \le q,
    \\
    0,
      \quad \! \text{when } p-\ell \ge q, \\
    \pi q^2,
      \quad \! \text{when } \ell-p \ge q,
  \end{cases}
  \nonumber
\end{align}
%\end{fleqn}
for $p$, $q$, $\ell>0$. Since $tv \, - \, \hat{a}r \ge \hat{a} w_{r, tv, \theta}$ never occurs for $\hat{a} > 1$, substituting \eqref{eq:cc_intersect_comp} into \eqref{eq:ccdf_TS0} and differentiating both sides of \eqref{eq:ccdf_TS0} with respect to $t$ lead to the result.
%yield the conditional PDF of $\tilde{T}_{S, \mathsf{f}(S)}$ given by the right-hand side of \eqref{eq:pdf_TS0}. 
%The result then follows from \eqref{eq:T_tildeT_equiv}.
\end{proof}

%\subsubsection{The first neighbor cell distribution}
\subsection{Target Cell Distribution}

We represent the BS location of the target cell $\mathsf{t}(I)$ relative to the location 
$\bm{u}_{T_{I, \mathsf{t}(I)}} = (T_{I, \mathsf{t}(I)} v, 0)$
as
\begin{equation}\label{eq:polar_tI_at_T}
  X_{\mathsf{t}(I)} - \bm{u}_{T_{I, \mathsf{t}(I)}} = (R_{\mathsf{t}(I)}, \Theta_{\mathsf{t}(I)}).
\end{equation}
The following result contributes to the characterization of
$X_{\mathsf{t}(I)}$.
%To characterize $X_{\mathsf{t}(I)}$, the following result is provided. 

%We define the location of the target cell $\mathsf{t}(I)$ relative to the UE at time $T_{I, \mathsf{t}(I)}$ by
%\[
%  X_{\mathsf{t}(I)} - \bm{u}_{T_{I, \mathsf{t}(I)}} = (R_{\mathsf{t}(I)}, \Theta_{\mathsf{t}(I)}).
%\]
%The following result characterizes this location.
%
%
%\begin{itemize}
%  \item
%  時刻$T_{I, \mathsf{t}(I)}$におけるUEを中心としたターゲットセル$\mathsf{t}(I)$の位置について、
%  \[ X_{\mathsf{t}(I)} - \bm{u}_{T_{I, \mathsf{t}(I)}} = (R_{\mathsf{t}(I)}, \Theta_{\mathsf{t}(I)}), \]
%  とおく。
%  $\mathsf{t}(I)$の位置を特徴づける次の結果を得る。
%\end{itemize}

\begin{proposition}\label{prop:pdf_TilTheta}
  Under the initial non-triggering event $\mathcal{N}_0$, given $X_I = (r, \theta)$ and $T_{I, \mathsf{t}(I)} = t > 0$,
%  such that $t$ lies in the support of the PDF given in \eqref{eq:pdf_TS0},
  it holds that
  \begin{equation}\label{eq:R_tS}
    R_{\mathsf{t}(I)} = \hat{r}(t),
  \end{equation}
  where $\hat{r}(t)$ is given in \eqref{eq:r_hat_t},
%  \begin{equation}\label{eq:r_hat_t}
%    \hat{r}(t) = \hat{a} w_{r, tv, \theta} = \hat{a} \sqrt{r^2 + t^2 v^2 - 2rtv\cos\theta}.
%  \end{equation}
  and the conditional PDF of $\Theta_{\mathsf{t}(I)}$ is approximately given by
  \begin{align}
    &f_{\Theta_{\mathsf{t}(I)} | \mathcal{N}_0}(\phi \,|\, r, \theta, t)
    \nonumber\\
    &\approx
    \frac
      {\cos\phi - \hat{a}\cos\varphi_{r, tv, \theta}}
      {2\sin\alpha_{r, tv, \theta} - 2\hat{a}\alpha_{r, tv, \theta}\cos\varphi_{r, tv, \theta}},
    \quad |\phi| < \alpha_{r, tv, \theta},
    \label{eq:dens_Theta_tS}
  \end{align}
  where
  \begin{align}
    \varphi_{r, tv, \theta} 
    &= \cos^{-1}
    \left(
      \frac{r\cos\theta - tv}{w_{r, tv, \theta}}
    \right),
    \label{eq:vphi_cosinv}
    \\
    \alpha_{r, tv, \theta} 
    &= \cos^{-1}(-1 \vee \hat{a}\cos\varphi_{r, tv, \theta}).
    \label{eq:alpha_cosinv}
  \end{align}
  with $x \vee y = \max(x, y)$.
  %, and $w_{r, tv, \theta} = \sqrt{r^2 + t^2 v^2 - 2rtv\cos\theta}$.
\end{proposition}
\begin{proof}
See Appendix~\ref{asec:pdf_TilTheta}.
\end{proof}

\begin{remark}
The introduced approximation allows the point $X_{\mathsf{t}(I)}$ to lie inside $b_{\bm{o}} \bigl(\hat{a}\|X_I\| \bigr)$ under $\mathcal{N}_0$, which never occurs if $\mathsf{Off} = 0$, \text{i.e.}, $\hat{a} = 1$.
Therefore, this approximation becomes exact as $\mathsf{Off} \to 0$.
The minimizer in \eqref{eq:ho_time_cell_fin} is almost surely unique under $\mathcal{N}_0$, since the probability of multiple points vanishes, as shown in Appendix~\ref{asec:pdf_TilTheta}.
\end{remark}

\begin{remark}%\label{rmk:baseline_model}
Previous studies such as \textnormal{\cite{wei2023equivalent}}
and \textnormal{\cite{wei2024time}}
model the conditional distribution of the target cell angle,
conditioned on the target cell distance,
as uniform over the feasible angular region.
Under the notation of Proposition~\ref{prop:pdf_TilTheta}, this baseline model is represented as
\begin{align}
  f^{\textnormal{base}}_{\Theta_{\mathsf{t}(I)} | \mathcal{N}_0}(\phi \,|\, r, \theta, t)
  =
  \frac
    {1}
    {2\alpha_{r, tv, \theta}},
  \quad |\phi| < \alpha_{r, tv, \theta},
  \label{eq:dens_Theta_tS_baseline}
\end{align}
which does not rely on the target cell definition in \eqref{eq:ho_time_cell_fin}.
A comparison of the HO performance obtained with the baseline and proposed models is presented in Sec.~\ref{sec:numerical_evaluation}.
%Comparison of the HO performance obtained with the baseline and proposed models is seen in Sec.~\ref{sec:numerical_evaluation}.
\end{remark}

\begin{figure}[!t]
    \centering
    \includegraphics[width=\linewidth]{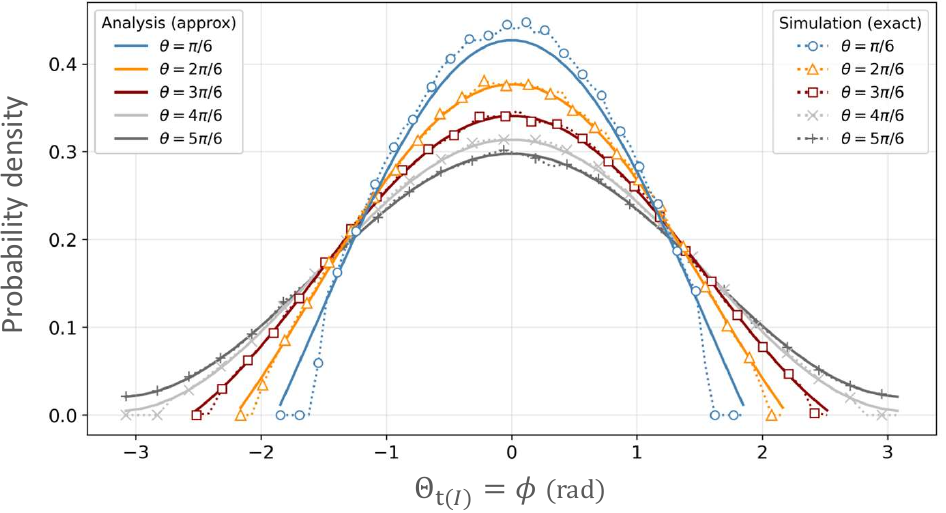}
%    \\ (i) $\hat{P}_{i, j} < 1$ \\
%    \vspace{3mm}
    \caption{
      Conditional PDF of the target cell angle $\Theta_{\mathsf{t}(I)}$ given $X_I = (r, \theta)$ and $T_{I, \mathsf{t}(I)} = t$. The analytical approximation in \eqref{eq:dens_Theta_tS} is compared with exact Monte Carlo simulations for different values of $\theta$.
    }
    \label{fig:1_prob_dens}
\end{figure}

Fig.~\ref{fig:1_prob_dens} compares the approximate  analytical results given by \eqref{eq:dens_Theta_tS} and the results from Monte Carlo simulations without the approximation for $r=0.5$ km, $t=50$ sec, $v=0.01$ km/sec, $\lambda=1$ km$^{-2}$, $\mathsf{Off}=-3$ dB, and $\beta=4$. Although a noticeable discrepancy is observed for small values of $\theta$, the analytical approximation captures the conditional distribution of $\Theta_{\mathsf{t}(I)}$ with good accuracy overall.

\section{Handover Event Probabilities}

In this section, we analyze the probabilities of the too-late HO and the ping-pong HO events described in Sec.~\ref{subsec:ho_event}.

\subsection{Approximation for Downlink Radio Link Quality}

Define the event $\mathcal{A}$ as
\begin{equation}\label{eq:set_conditions}
  \mathcal{A}
  =
  \left\{
    X_I = (r, \theta), \, T_{I, \mathsf{t}(I)} = t, \, X_{\mathsf{t}(I)} - \bm{u}_{t} = (\hat{r}(t), \phi)
  \right\}.
\end{equation}
where $r > 0$, $\theta \in (-\pi, \pi)$, $t > 0$, and $\phi \in (-\alpha_{r \!, tv \!, \theta}, \alpha_{r \!, tv \!, \theta})$.
By \eqref{eq:polar_tI_at_T} and \eqref{eq:R_tS}, the above definition is equivalent to
\begin{equation*}
\tag{\ref{eq:set_conditions}$'$}
\label{eq:set_conditions_prime}
  \mathcal{A}
  =
  \{
    X_I = (r, \theta), T_{I, \mathsf{t}(I)} = t, \Theta_{\mathsf{t}(I)} = \phi
  \}.
\end{equation*}
%
%For $s \ge 0$, the downlink radio link quality $\rho(s \mid X_I, X_{\mathsf{t}(I)})$ under $\mathcal{A}$ is
%denoted by $\rho(s \mid \mathcal{A})$.
For $s \ge 0$, under $\mathcal{A}$,
the downlink radio link quality
$\rho(s \mid X_I, X_{\mathsf{t}(I)})$ in \eqref{eq:downlink_RLQ}
is abbreviated as $\rho(s \mid \mathcal{A})$.

\begin{lemma}\label{lem:cond_expect_SINR}
The downlink radio link quality under $\mathcal{A}$ satisfies
\begin{align}\label{eq:cond_expect_SINR}
  &\rho(s \mid \mathcal{A})
  \nonumber\\
  &=
  \int_0^\infty \!\!
    \frac{ w_{r \!, sv \!, \theta}^{-\beta} }{ 1 + z\, w_{\hat{r}(t), (s-t)v, \phi}^{-\beta}}
%    \exp
%    \left(
    \ee^{
      -\sigma^2 z - \pi\lambda K_\beta \,z^{2/\beta}
    }
%    \right)
    \mu(z, s, t)
  \dd z,
\end{align}
where $w_{r \!, sv \!, \theta} = \sqrt{r^2 + s^2v^2 - 2rsv\cos\theta}$, 
%$\hat{r}(t)$ is given in \eqref{eq:r_hat_t}, 
and
\begin{align}
  K_\beta 
  &= 
  \frac{2\pi}{\beta}\csc\frac{2\pi}{\beta},
  \label{eq:K_beta}\\
  \mu(z, s, t) 
  &=
  \exp\left(
    2\lambda z
    \int_0^{tv + \hat{r}(t)} \!\!\!\!
      \int_0^{C(x, \hat{r}(t), tv)} \!
        \frac
          {x}
          {z + w_{x \!, sv \!, \psi}^{\,\beta}}
      \dd \psi
    \dd x
  \right),
  \label{eq:func_mu_st}
\end{align}
with
\begin{align*}
  &C(x, \hat{r}(t), tv)
  \\
  &=
  \begin{cases}
    \pi, & 0 \le x \le \hat{a}r, 
  \\ 
    \cos^{-1} \left(-1 \vee \frac{x^2 + t^2v^2 - \hat{r}(t)^2}{2 x tv} \right), & \hat{a}r \le x \le tv + \hat{r}(t).
  \end{cases}
\end{align*}
\end{lemma}
\begin{proof}
%See Appendix~\ref{asec:lem_cond_expect_SINR} of the technical paper~\cite{tokuyama2026technical}.
See Appendix~\ref{asec:lem_cond_expect_SINR}.
\end{proof}

To obtain a tractable analytical characterization of the RLF event,
we adopt the following approximation.
\begin{approximation}\label{approx:cond_expect_SINR}
The function $\rho(s \mid \mathcal{A})$ in
\eqref{eq:cond_expect_SINR} is treated as monotonically decreasing
or unimodal with respect to $s \ge 0$.
\end{approximation}
The validity and impact of this approximation are numerically
investigated in Appendix~\ref{asec:assump_cond_expect_SINR}.

\subsection{Results for HO Event Probabilities}

This subsection presents analytical results for the probabilities of the events
$\mathcal{L}$ and $\mathcal{P}$ defined in \eqref{eq:ev_too_late} and
\eqref{eq:ev_ping_pong}. We first derive the expression for
$\Prb(\mathcal{P})$, and then obtain the result for $\Prb(\mathcal{L})$ as its corollary.

In this paper, the Heaviside step function $H(x)$ is defined as 
%
%$H(x)=1$ for $x\ge0$ and $H(x)=0$ for $x<0$.
%
%
\[
H(x)=
\biggl\{
\begin{array}{ll}
1, & x \ge 0,\\
0, & x < 0.
\end{array}
\biggl.
\]
%In this paper, the Heaviside step function is defined as
%\[
%H(x)=
%\begin{cases}
%1, & x\ge0,\\
%0, & x < 0.
%\end{cases}
%\]

%In this paper, the Heaviside step function and its shifted version are defined as
%\[
%H(x)=
%\begin{cases}
%1, & x \ge 0,\\
%0, & x < 0,
%\end{cases}
%\qquad
%H_a(x) := H(x-a),
%\]
%where $a \in \mathbb{R}$ denotes the shift parameter.

\begin{theorem}\label{thm:ping-pong}
The probability of the ping-pong HO event is
\begin{align*}%\label{eq:ping-pong_result}
&\Prb(\mathcal{P})
=
\lambda
\int_0^\infty \!\!\!
\int_0^\pi \!\!\!
\int_0^\infty \!\!\!
\int_{-\alpha_{r \!, tv \!, \theta}}^{\alpha_{r \!, tv \!, \theta}} 
r \ee^{-\pi\lambda\hat{a}^2r^2} 
\, \eta(\phi, t, r, \theta)
\nonumber\\
&\quad\times
\Bigl[
1 -
H\bigl(
  \min\{ 2\Delta_{\mathsf{T}} v, q_{\! \mathcal{A}}^+ \}
  -
  \max\{ \Delta_{\mathsf{T}} v, q_{\! \mathcal{A}}^- \}
\bigr)
H\bigl(d_{\! \mathcal{A}} - y_{\! \mathcal{A}}\bigr)
\Bigr]
\nonumber\\
&\quad\times
%
%\nonumber\\
%&\quad\times
H\bigl(
%\max\{
%  \zeta_{\mathcal{A}}(t + \Delta_{\mathsf{T}} - \Delta_{\mathsf{F}}),\;
%  \zeta_{\mathcal{A}}(t + \Delta_{\mathsf{T}})
%\}
J_{\mathcal{A}}(t)
-
Q_{\mathsf{F}}
\bigr)
H\bigl(
%\min\{
%  \zeta_{\mathcal{A}}(0),\;
%  \zeta_{\mathcal{A}}(t + \Delta_{\mathsf{T}} - \Delta_{\mathsf{F}})
%\}
%
\zeta_{\mathcal{A}}(0)
-
Q_{\mathsf{F}}
\bigr)
\, \dd \phi\, \dd t\, \dd \theta\, \dd r.
\end{align*}
where
%where $H(x)$ is the Heaviside step function, which we define $H(x) = 1$ when $x > 0$ and $H(x) = 0$ otherwise, and
%where $H(x)$ is the Heaviside step function defined by $H(x)=1$ for $x>0$ and $H(x)=0$ otherwise, and
%
\begin{align}
  \eta(\phi, t, r, \theta)
  &=
  2 f_{T_{I, \mathsf{t}(I)} | \mathcal{N}_0}(t \,|\, r, \theta)
  \,
  f_{\Theta_{\mathsf{t}(I)} \mid \mathcal{N}_0, X_I, T_{I,\mathsf{t}(I)}}(\phi \mid r,\theta,t)
  \nonumber\\
  &\approx
  \frac
    {\cos\phi - \hat{a}\cos\varphi_{r \!, tv \!, \theta}}
    {\sin\alpha_{r \!, tv \!, \theta} - \hat{a}\alpha_{r \!, tv \!, \theta}\cos\varphi_{r \!, tv \!, \theta}}
  f_{T_{I, \mathsf{t}(I)} | \mathcal{N}_0}(t \,|\, r, \theta),
  \label{eq:func_eta}\\
  J_{\mathcal{A}}(t) 
  %&= \max_{-\Delta_{\mathsf{F}} \le s \le 0} \zeta_{\mathcal{A}}(t + \Delta_{\mathsf{T}} + s)
  %\nonumber\\
  &\approx 
  %\max\{
  %  \zeta_{\mathcal{A}}(t + \Delta_{\mathsf{T}} - \Delta_{\mathsf{F}}),\;
  %  \zeta_{\mathcal{A}}(t + \Delta_{\mathsf{T}})
  %\}
  \zeta_{\mathcal{A}}(t + \Delta_{\mathsf{T}} - \Delta_{\mathsf{F}}),
  \label{eq:func_J_A}
  \\
  \zeta_{\mathcal{A}}(s)
  &=
  \begin{cases}
    \rho(s \mid \mathcal{A}), & s \ge 0, 
  \\ 
    Q_{\mathsf{F}}, & s < 0,
  \end{cases}
  \label{eq:func_zeta}
  %\\
  %q_{\! \mathcal{A}}^\pm
  %&=
  %x_{\! \mathcal{A}} \pm \sqrt{d_{\! \mathcal{A}}^2 - y_{\! \mathcal{A}}^2},
  %\nonumber
  %\label{eq:intersection_q_pm}
\end{align}
with $f_{T_{I, \mathsf{t}(I)} | \mathcal{N}_0}(t \,|\, r, \theta)$ and $\rho(s \mid \mathcal{A})$ given in \eqref{eq:pdf_TS0} and \eqref{eq:cond_expect_SINR}, respectively.
Furthermore, 
$
q_{\! \mathcal{A}}^\pm
=
x_{\! \mathcal{A}} \pm \sqrt{d_{\! \mathcal{A}}^2 - y_{\! \mathcal{A}}^2}
$, 
where 
%$x_{\! \mathcal{A}}$, $y_{\! \mathcal{A}}$, and $d_{\! \mathcal{A}}$ are given by
%
\begin{align}\label{eq:center_apollonius}
  &(x_{\! \mathcal{A}}, y_{\! \mathcal{A}})
  \nonumber\\
  &=
  \left(
    \frac
      {\hat{a}^3\cos\phi - \cos\varphi_{r \!, tv \!, \theta} }
      {\hat{a}^2 - 1}
    w_{r \!, tv \!, \theta},
    \frac
      {\hat{a}^3\sin\phi - \sin\varphi_{r \!, tv \!, \theta} }
      {\hat{a}^2 - 1}
    w_{r \!, tv \!, \theta}
  \right),
\end{align}
and
\begin{equation}\label{eq:radius_apollonius}
  d_{\! \mathcal{A}}
  =
  \frac
    {\hat{a}w_{r \!, tv \!, \theta}}
    {\hat{a}^2 - 1}
  \sqrt{\hat{a}^2 - 2\hat{a}\cos(\phi - \varphi_{r \!, tv \!, \theta}) + 1},
\end{equation}
with $w_{r \!, tv \!, \theta} = \sqrt{r^2 + t^2v^2 - 2rtv\cos\theta}$, and $\varphi_{r \!, tv \!, \theta}$ is in \eqref{eq:vphi_cosinv}. 
\end{theorem}

\begin{proof}
See Appendix~\ref{asec:ping-pong}.
\end{proof}

\begin{corollary}\label{cor:too-late}
%\mathindent=0mm
%
The probability of the too-late HO event is
\begin{align*}%\label{eq:too-late_result}
&\Prb(\mathcal{L})
=
\lambda
\int_0^\infty \!\!\!
\int_0^\pi \!\!\!
\int_0^\infty \!\!\!
\int_{-\alpha_{r \!,tv \!, \theta}}^{\alpha_{r \!, tv \!, \theta}}
r \ee^{-\pi\lambda\hat{a}^2 r^2}
\, \eta(\phi, t, r, \theta)
\nonumber\\
&\quad\times
\bar{H}\bigl(
  J_{\mathcal{A}}(t) 
  - Q_{\mathsf{F}}
\bigr)
%\,
H\bigl(
  \zeta_{\mathcal{A}}(0) - Q_{\mathsf{F}}
\bigr)
\,\dd \phi\,\dd t\,\dd \theta \,\dd r.
\end{align*}
where $\bar{H}(x) = 1 - H(x)$, and $\eta(\phi, t, r, \theta)$ and $J_{\mathcal{A}}(t)$ are approximately given in \eqref{eq:func_eta} and \eqref{eq:func_J_A}, respectively.
\end{corollary}

\begin{proof}
See Appendix~\ref{asec:too-late}.
\end{proof}

\section{Numerical Evaluation}\label{sec:numerical_evaluation}

\begin{figure}[!t]
    \centering
    \includegraphics[width=\linewidth]{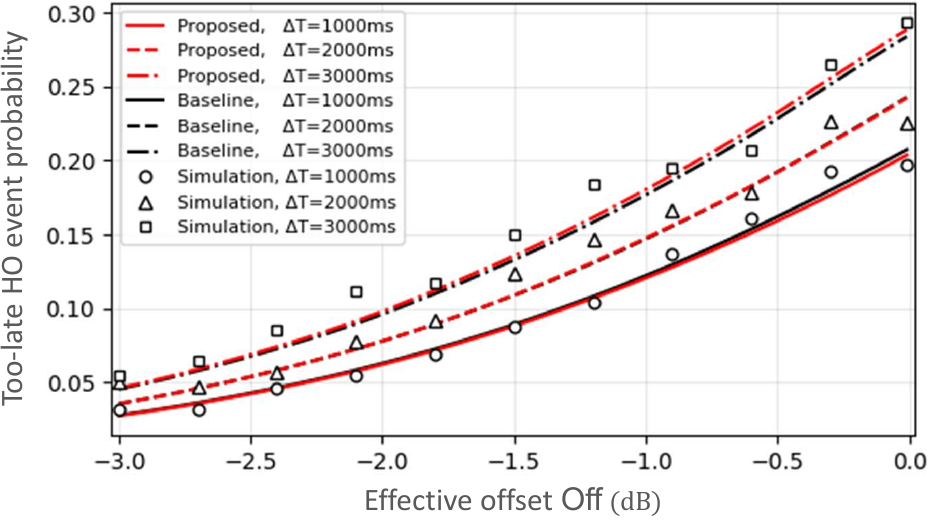}
    \\ (a) $\Prb(\mathcal{L})$. \\
    \vspace{3mm}
    \includegraphics[width=0.99\linewidth]{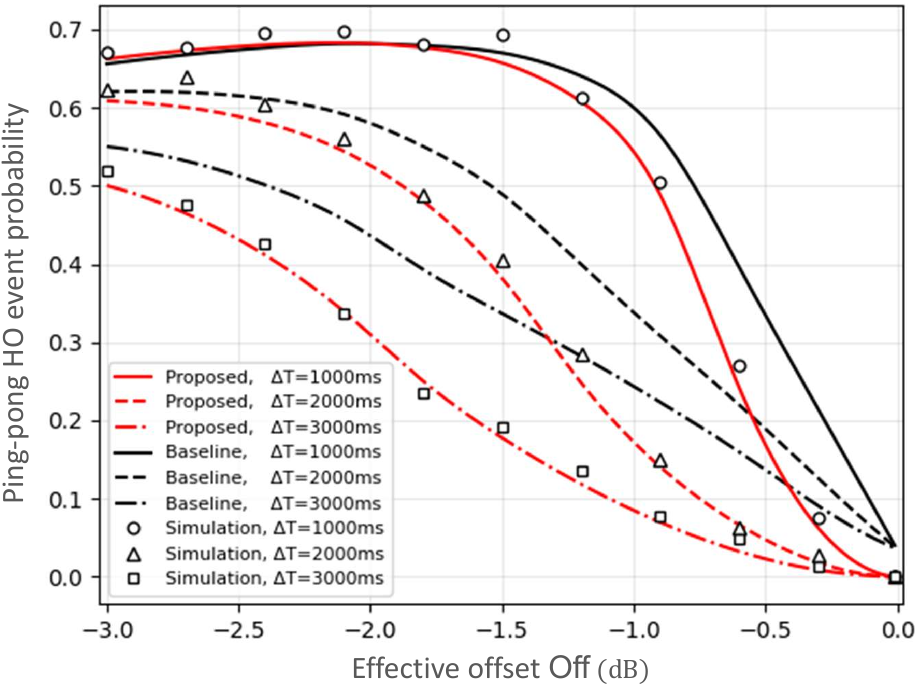}
    \\ (b) $\Prb(\mathcal{P})$.
    \caption{
    Too-late HO and ping-pong HO probabilities versus the effective
    offset $\mathsf{Off}$ for different TTT durations
    $\Delta_{\mathsf{T}}$: proposed analysis, baseline model, and
    Monte Carlo simulation ($\lambda=3~\mathrm{km}^{-2}$).
    }
    \label{fig:2_ho_prob}
\end{figure}

We numerically validate the proposed analytical framework through
comparisons with Monte Carlo simulations and a conventional target-cell
model, and investigate the tradeoff between the too-late and ping-pong
HO probabilities with respect to the TTT duration.
Unless otherwise stated, the parameters are
$v=0.01~\mathrm{km/sec}$, $\beta=4$, $\sigma^2=0$,
$\Delta_{\mathsf{F}}=3000~\mathrm{ms}$, and
$Q_{\mathsf{F}}=0.5$.
The BS density is selected from
$\lambda=1$, $3$, and $10~\mathrm{km}^{-2}$,
the TTT duration from
$0<\Delta_{\mathsf{T}}\leq5000~\mathrm{ms}$,
and the negative effective offset from
$-3\leq\mathsf{Off}<0~\mathrm{dB}$.
These TTT and offset ranges are consistent with the 3GPP specification
\cite{3gppTS36331} and commonly adopted HO margins
\cite{altrad2014doppler}, respectively.

\subsection{Comparison with Simulation and Previous Model}
%\label{sebsec:comp_sim_prevmodel}

Fig.~\ref{fig:2_ho_prob} compares the proposed analytical results with
Monte Carlo simulations and the baseline model for
$\lambda=3~\mathrm{km}^{-2}$ and
$\Delta_{\mathsf{T}}=1000$, $2000$, and $3000~\mathrm{ms}$.
The proposed results are obtained from
Corollary~\ref{cor:too-late} and Theorem~\ref{thm:ping-pong},
whereas the baseline results are obtained by replacing
$f_{\Theta_{\mathsf{t}(I)} \mid \mathcal{N}_0, X_I, T_{I,\mathsf{t}(I)}}$
in \eqref{eq:func_eta} with the baseline density
\eqref{eq:dens_Theta_tS_baseline}.
The Monte Carlo results for
\eqref{eq:ev_too_late} and \eqref{eq:ev_ping_pong}
are averaged over 1000 independent PPP realizations, with the
conditional expectation in \eqref{eq:downlink_RLQ} estimated from
10000 independent realizations of the fading and inner PPP.

For both HO events, the proposed analytical results closely match the
simulations over the considered effective offsets and TTT durations,
despite the approximation errors observed in
Fig.~\ref{fig:1_prob_dens}.
The baseline model also captures the overall trends but exhibits
noticeable discrepancies, particularly for the ping-pong HO probability
at moderately negative offsets
($-1.5~\mathrm{dB}\lesssim\mathsf{Off}\lesssim-0.5~\mathrm{dB}$).
In this region, the selected target cell in
\eqref{eq:ho_time_cell_fin} has a markedly non-uniform angular
distribution, demonstrating the importance of explicitly modeling the
target BS distribution.

\subsection{Joint HO Events and Optimal TTT Timer}
%\label{subsec:joint_events_opt_TTT}

Fig.~\ref{fig:3_tradeoff} shows
$\Prb(\mathcal{L})$, $\Prb(\mathcal{P})$, and their sum
$\Prb(\mathcal{L})+\Prb(\mathcal{P})
=\Prb(\mathcal{L}\cup\mathcal{P})$
versus $\Delta_{\mathsf{T}}$ for
$\lambda=1$, $3$, and $10~\mathrm{km}^{-2}$ and
$\mathsf{Off}=-1~\mathrm{dB}$.
The equality follows from the mutual exclusivity of
$\mathcal{L}$ and $\mathcal{P}$.

As $\Delta_{\mathsf{T}}$ increases, the too-late HO probability
increases while the ping-pong HO probability decreases, illustrating
their inherent tradeoff.
Their sum exhibits a distinct minimum for each BS density, indicating
an optimal TTT duration.
Moreover, the optimum shifts toward smaller TTT values as the BS
density increases because both probabilities vary more rapidly with
TTT in denser deployments.
Thus, the optimal TTT duration depends on the BS density and should be
adapted to the deployment environment.

\begin{figure}[!t]
    \centering
    \includegraphics[width=\linewidth]{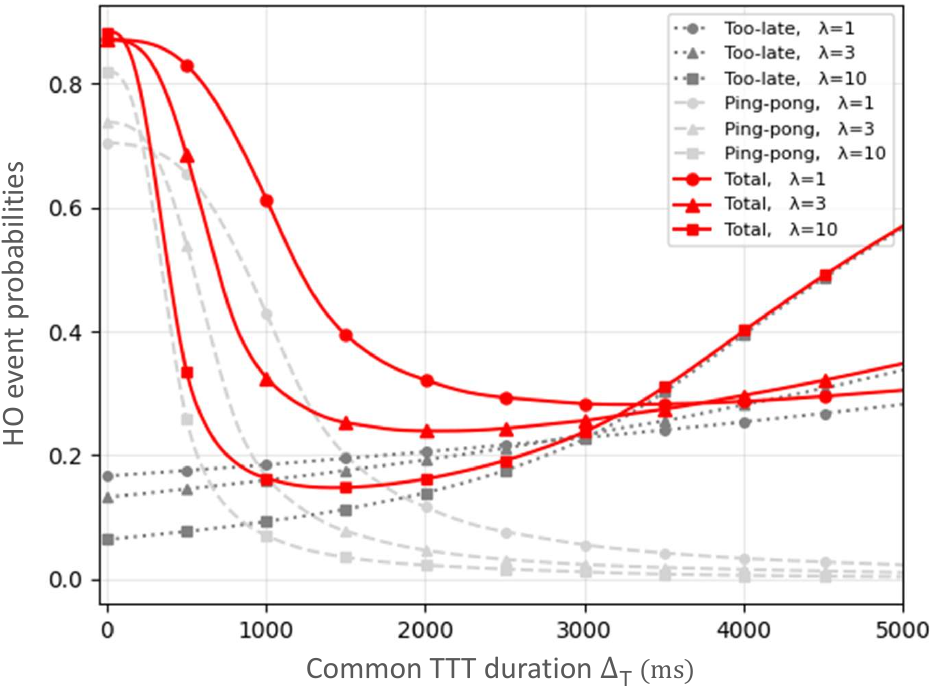}
    \caption{
    Probabilities $\Prb(\mathcal{L})$, $\Prb(\mathcal{P})$, and their
    sum versus the TTT duration $\Delta_{\mathsf{T}}$ for different
    BS densities $\lambda$ $\mathrm{km}^{-2}$.
    }
    \label{fig:3_tradeoff}
\end{figure}

\section{Conclusion}

This paper developed a stochastic geometry-based framework for mobility
robustness analysis by explicitly modeling the spatial distribution of
the target BS induced by UE mobility.
Based on the derived target BS distribution, we analytically evaluated
the too-late and ping-pong HO probabilities under the SINR model with
aggregate interference.
Numerical results demonstrated improved accuracy over the baseline model
considering a uniform target-BS angle and revealed the TTT-dependent
tradeoff between the two HO events and its dependence on BS density.
Future work includes extensions to heterogeneous and UAV-assisted
networks, positive effective offsets, and exact characterization of the
target BS distribution.

%% remaining page estimation test %%
%hello hello hello hello hello hello hello hello hello hello hello hello hello hello hello hello hello hello hello hello hello hello
%hello hello hello hello hello hello hello hello hello hello hello hello hello hello hello hello hello hello hello hello hello hello 
%hello hello hello hello hello hello hello hello hello hello hello hello hello hello hello hello hello hello hello hello hello hello 
%hello hello hello hello hello hello hello hello hello hello hello

\appendix

\subsection{Numerical Justification of Approximation~\ref{approx:cond_expect_SINR}}
\label{asec:assump_cond_expect_SINR}

In this subsection, we numerically investigate the validity of Approximation~\ref{approx:cond_expect_SINR} and its impact on RLF detection and the resulting HO event probabilities.

Table~\ref{tab:summary} summarizes the numerical results. For each $\mathsf{Off}\in\{-1,-2,-3,-4,-5\}~\mathrm{dB}$, we generate $10{,}000$ samples of ($r$, $\theta$), $t$, and $\phi$ according to \eqref{eq:joint_pdf_XI}, \eqref{eq:pdf_TS0}, and \eqref{eq:dens_Theta_tS}, respectively. The other parameters are set to $v=0.01~\mathrm{km/sec}$, $\lambda=1~\mathrm{km}^{-2}$, $\beta=4$, and $\sigma=0$.
For each sample, $\rho(s\mid\mathcal{A})$ is evaluated up to
$s=200~\mathrm{sec}$ with a step size of $100~\mathrm{ms}$.
We examine whether the condition underlying
Approximation~\ref{approx:cond_expect_SINR} holds and whether applying
the approximation changes the RLF detection result.
To evaluate its impact on RLF detection over a range of radio link
quality thresholds, we consider
$Q_{\mathsf{F}}\in\{0.10,0.11,\ldots,1.00\}$ and set
$\Delta_{\mathsf{F}}=3000~\mathrm{ms}$.
For each sample that does not satisfy the condition underlying
Approximation~\ref{approx:cond_expect_SINR}, the RLF detection result
is regarded as changed if, for at least one value of $Q_{\mathsf{F}}$,
$\rho(s\mid\mathcal{A})$ falls below $Q_{\mathsf{F}}$ and subsequently
recovers to $Q_{\mathsf{F}}$ or above within the duration of $\Delta_{\mathsf{F}}$.

As shown in Table~\ref{tab:summary}, the condition underlying
Approximation~\ref{approx:cond_expect_SINR} fails for $27.5$--$32.2\%$
of the samples. Nevertheless, the RLF detection result is changed by
the approximation for only $0.5$--$0.8\%$ of the samples, even when
considering the range $Q_{\mathsf{F}}\in[0.1,1.0]$.
This indicates that a violation of the condition underlying
Approximation~\ref{approx:cond_expect_SINR} does not necessarily lead
to an error in RLF detection. This is because deviations from
monotonic or unimodal behavior affect RLF detection only when they
cause the radio link quality to cross a threshold and subsequently
recover above it within the interval of duration
$\Delta_{\mathsf{F}}$.

The limited impact of Approximation~\ref{approx:cond_expect_SINR} is also observed in the resulting HO event probabilities. The results for $\lambda=3~\mathrm{km}^{-2}$ in Fig.~\ref{fig:2_ho_prob} show that the proposed analysis reasonably captures both $\Prb(\mathcal{L})$ and $\Prb(\mathcal{P})$ obtained by Monte Carlo simulation. Fig.~\ref{afig:ho_prob} further shows a similar agreement for $\lambda=1~\mathrm{km}^{-2}$ over different values of $\mathsf{Off}$ and $\Delta_{\mathsf{T}}$. These results indicate that, although the condition underlying Approximation~\ref{approx:cond_expect_SINR} does not always hold on a sample-by-sample basis, its impact on RLF detection and the resulting HO event probabilities remains limited, thereby providing numerical support for the use of Approximation~\ref{approx:cond_expect_SINR}.

\begin{table}[t]
\centering
\caption{Validity of Approximation~\ref{approx:cond_expect_SINR} and its impact on RLF detection.}
\label{tab:summary}
\begin{tabular}{c|cc|cc}
\hline
 & \multicolumn{2}{c|}{Condition for Approx.~\ref{approx:cond_expect_SINR}}
 & \multicolumn{2}{c}{Impact on RLF detection} \\
 & Holds & Fails
 & Unchanged & Changed \\
\hline
$\mathsf{Off} = -1$ & 67.8\% & 32.2\% & 99.2\% & 0.8\% \\
$\mathsf{Off} = -2$ & 68.5\% & 31.5\% & 99.4\% & 0.6\% \\
$\mathsf{Off} = -3$ & 70.6\% & 29.4\% & 99.5\% & 0.5\% \\
$\mathsf{Off} = -4$ & 71.5\% & 28.5\% & 99.2\% & 0.8\% \\
$\mathsf{Off} = -5$ & 72.5\% & 27.5\% & 99.2\% & 0.8\% \\
\hline
\end{tabular}
\end{table}

\begin{figure}[!t]
    \centering
    \includegraphics[width=\linewidth]{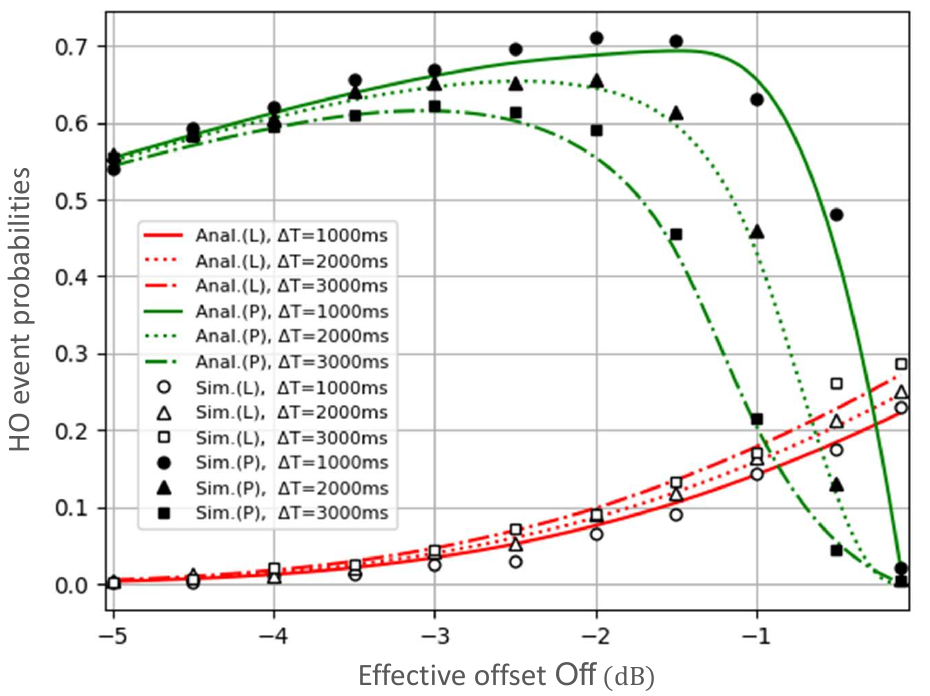}
    \caption{
    Too-late HO and ping-pong HO probabilities versus the effective
    offset $\mathsf{Off}$ for different TTT durations
    $\Delta_{\mathsf{T}}$: proposed analysis and Monte Carlo simulation ($\lambda=1~\mathrm{km}^{-2}$).
    }
    \label{afig:ho_prob}
\end{figure}

%%% コメントアウト（Appendix~B）―開始― %%%
%\begin{comment}

\subsection{Supplementary Proofs}

\subsubsection{Proof of Proposition~\ref{prop:pdf_TilTheta}} \label{asec:pdf_TilTheta}

On the event $T_{I,\mathsf{t}(I)} \in (t,t+\delta_t]$ for $t > 0$, 
let $R_{\mathsf{t}(I)}(\bm{u}_t)$ and $\Theta_{\mathsf{t}(I)}(\bm{u}_t)$ denote the magnitude and polar angle of 
$X_{\mathsf{t}(I)}-\bm{u}_t$, respectively.
We then consider their distributions, which converge as $\delta_t\to0$ to those of 
$R_{\mathsf{t}(I)}$ and $\Theta_{\mathsf{t}(I)}$, 
since the displacement between $\bm{u}_t$ and $\bm{u}_{T_{I,\mathsf{t}(I)}}$ vanishes as $\delta_t\to0$.

%
%Let $\delta_t > 0$ be sufficiently small.
Define the conditional void probability of the homogeneous PPP $\Phi$ on $B \subset \R^2$ as
\begin{equation*}
  \mathcal{V}(B \mid r, \theta) = \Prb\bigl(\Phi(B) = 0 \mid \mathcal{N}_0, X_I = (r, \theta)\bigr).
\end{equation*}
From the second equality of \eqref{eq:ccdf_TS0}, we have
\begin{align*}
  &\Prb\bigl( T_{I, \mathsf{t}(I)} \in (t, t + \delta_t] \mid \mathcal{N}_0, X_I = (r, \theta) \bigr) 
  \nonumber\\
  &=
  \mathcal{V}
  \Bigl(
    b_{\bm{o}} \bigl(\hat{a} \|X_I\| \bigr)
    \cup 
    b_{\bm{u}_{t}} \bigl(\hat{a} \|X_I - \bm{u}_{t}\| \bigr) 
    \,\, \big| \,\, r, \theta
  \Bigr)
%  \nonumber\\
%  &\quad
%    \times
    \bar{\mathcal{V}}
    \bigl(
      D(\delta_t)
      \, \big| \,\, r, \theta
    \bigr),
\end{align*}
where $\bar{\mathcal{V}}(B \mid r, \theta)$ is defined in \eqref{eq:comp_void_prob_PPP}, and the region $D(\delta_t)$ is
\begin{equation}\label{eq:D_delta_ell}
  D(\delta_t)
  =
  b_{\bm{u}_{t + \delta_t}} \bigl(\hat{a}\|X_I - \bm{u}_{t + \delta_t}\|\bigr)
  \big\backslash 
  b_{\bm{o}} \bigl(\hat{a}\|X_I\| \bigr)
  \big\backslash 
  b_{\bm{u}_{t}} \bigl(\hat{a}\|X_I - \bm{u}_{t}\| \bigr).
\end{equation}
Therefore, given $\mathcal{N}_0$, $X_I = (r, \theta)$, and $T_{I, \mathsf{t}(I)} \in (t, t + \delta_t]$, the point $X_{\mathsf{t}(I)} \in \Phi$ lies in
\begin{fleqn}
\begin{align}\label{eq:D_delta_ell_approx}
  D(\delta_t)
  \approx
  \tilde{D}(\delta_t)
%  &=
%  b_{\bm{u}_{t + \delta_t}} \bigl(\hat{a} \|X_I - \bm{u}_{t + \delta_t}\|\bigr)
%  \backslash 
%  b_{\bm{u}_{t}} \bigl(\hat{a} \|X_I - \bm{u}_{t}\| \bigr)
%  \nonumber\\
  =
  b_{(tv +\delta_t v, 0)}(\hat{a} w_{r \!, tv +\delta_t v \!, \theta})
  \backslash
  b_{(tv, 0)}(\hat{a} w_{r \!, tv \!, \theta}),
\end{align}
\end{fleqn}
where $w_{r \!, tv \!, \theta} = \sqrt{r^2 + t^2v^2 - 2rtv\cos\theta}$. 
The approximate region $\tilde{D}(\delta_t)$ is obtained by ignoring the subtraction of $b_{\bm{o}}(\hat{a} r)$.

Let $\bm{z} \in \tilde{D}(\delta_t)$, and set $\bm{\ell} = (tv, 0)$, and $\bm{\ell}^+ = (tv + \delta_t v, 0)$.
Let $l_{\bm{z}}$ denote the half-line emanating from $\bm{\ell}$ and passing through $\bm{z}$, and write $\bm{z} - \bm{\ell} = (r_{\bm{z}}, \theta_{\bm{z}})$.
%Note that $|\theta_{\bm{z}}| < \varphi_{r \!, \ell, \theta}$, where $\varphi_{r \!, \ell, \theta}$ is given in \eqref{eq:vphi_cosinv}.
%
Let $\bm{z}_0$ and $\bm{z}^+$ denote the inner and outer intersection points of $l_{\bm{z}}$ with the boundary $\partial \tilde{D}(\delta_t)$, respectively, and define $h_{\theta_{\bm{z}}}\!(\delta_t) = \|\bm{z}^+ - \bm{z}_0\|$.
%, and define $\|\bm{z}_0 - \bm{\ell}\| = r_0$ and $\|\bm{z}^+ - \bm{\ell}\| = r_0 + h_{\theta_z}(\delta_\ell)$.
Due to \eqref{eq:D_delta_ell_approx},
\begin{align*}
\|\bm{z}^+ - \bm{\ell}^+\| = \hat{a} w_{r \!, tv + \delta_t v \!, \theta},
\quad
%\nonumber\\
\|\bm{z}_0 - \bm{\ell}\| =  \hat{a} w_{r \!, tv \!, \theta}.
\end{align*}
%
%Since the three points $\bm{\ell}$, $\bm{\ell}^+$, and $\bm{z}^+$ form a triangle, we have
Using the law of cosines applied to the triangle with vertices
$\bm{z}^+$, $\bm{\ell}$, and $\bm{\ell}^+$,
with angle $|\theta_{\bm{z}}|$,
and solving for $h_{\theta_{\bm{z}}}(\delta_t)$,
we obtain
%
%\begin{fleqn}
\begin{align}\label{eq:h_thetaz_sol}
&h_{\theta_{\bm{z}}}\!(\delta_t)
\nonumber\\
&=
- (\hat{a}w_{r \!, tv \!, \theta} - \delta_t v\cos\theta_{\bm{z}})
+ \sqrt{
  \hat{a}^2 w_{r \!, tv + \delta_t v \!, \theta}^2  - \delta_t^2 v^2 \sin^2\theta_{\bm{z}} 
}
\nonumber\\
&=\delta_t v(\cos\theta_{\bm{z}} - \hat{a}\cos\varphi_{r \!, tv \!, \theta}) + o(\delta_t), \quad \delta_t \rightarrow 0,
\end{align}
%\end{fleqn}
%
where we used $\hat{a}w_{r \!, tv \!, \theta} - \delta_t v\cos\theta_{\bm{z}} > 0$ for sufficiently small $\delta_t$, and the first-order expansion
\[
  w_{r \!, tv + \delta_t v \!, \theta} - w_{r \!, tv \!, \theta} = -\delta_t v\cos\varphi_{r \!, tv \!, \theta} + o(\delta_t),
\] 
with $\varphi_{r \!, tv \!, \theta}$ in \eqref{eq:vphi_cosinv}.
%
%
%
%Under $\tilde{T}_{S, \mathsf{f}(S)} \in (t, t + \delta_t]$, 
Since $t$ lies in the support of 
$f_{T_{I, \mathsf{t}(I)} | \mathcal{N}_0}(t \mid r, \theta)$ 
in \eqref{eq:pdf_TS0}, we have
$|\hat{a}r - tv| \le \hat{r}(t)$.
Using $\hat{r}(t) = \hat{a} w_{r \!, tv \!, \theta}$ from \eqref{eq:r_hat_t},
this is equivalent to
$\hat{a}r\cos\theta \le tvA^{-} + r$,
where $A^{-} = (\hat{a}^2 - 1)/(2\hat{a})$.
%
%Since $t$ lies in the support of $f_{T_{I, \mathsf{t}(I)} | \mathcal{N}_0}(t \mid r, \theta)$ in \eqref{eq:pdf_TS0}, we see that 
%$|\hat{a}r - tv| \le \hat{r}(t)$,
%which is equivalent to,
%by $\hat{r} = \hat{a} w_{r \!, tv \!, \theta}$ from \eqref{eq:r_hat_t},
%$\hat{a}r\cos\theta \le tvA^{\!-} + r$, with $A^{\!-} = (\hat{a}^2 - 1)/(2\hat{a})$.
%
Hence, using \eqref{eq:vphi_cosinv} and
$A^{\!-} < \hat{a} - 1$ for $\hat{a} > 1$,
we have
\[
  \hat{a} \cos\varphi_{r \!, tv \!, \theta} < 1.
\]
Since $\cos\theta_{\bm{z}} > \hat{a}\cos\varphi_{r \!, tv \!, \theta}$ from \eqref{eq:h_thetaz_sol}, 
it follows that
$
\theta_{\bm{z}} \in (-\alpha_{r \!, tv \!, \theta}, \alpha_{r \!, tv \!, \theta}),
$
where $\alpha_{r \!, tv \!, \theta}$ is given in \eqref{eq:alpha_cosinv}.

Let $|\cdot|$ denote the two-dimensional Lebesgue measure.
%, and let $\hat{r}(t)$ as given in \eqref{eq:r_hat_t}.
%$r_0 = \|\bm{z}_0 - \bm{\ell}\| = \hat{a}w_{r \!, tv \!, \theta}$.
%
%
Since $X_{\mathsf{t}(I)}$ lies in $D(\delta_t) \subset \tilde{D}(\delta_t)$,
 and $\|\bm{z}_0 - \bm{\ell}\| = \hat{a} w_{r \!, tv \!, \theta} = \hat{r}(t)$,
it follows that
\begin{equation*}
R_{\mathsf{t}(I)}(\bm{u}_t) \in [\hat{r}(t), \hat{r}(t) + h_{ \Theta_{\mathsf{t}(I)}\!(\bm{u}_t) }(\delta_t) ),
\end{equation*}
where $h_{ \Theta_{\mathsf{t}(I)}\!(\bm{u}_t) }(\delta_t)$ is given in \eqref{eq:h_thetaz_sol}.
Taking the limit $\delta_t \to 0$ in the above relation, we obtain \eqref{eq:R_tS}. 
As $\delta_t \to 0$, the probability of having more than one point 
in $D(\delta_t)$ defined in \eqref{eq:D_delta_ell} is $o(\delta_t)$ (see, e.g., Sec~2.3 of \cite{chiu2013stochastic}).
Hence, given $\mathcal{N}_0$, $X_I = (r, \theta)$, and 
$T_{I, \mathsf{t}(I)} \in (t, t + \delta_t]$, 
$X_{\mathsf{t}(I)}$ is asymptotically uniformly distributed over $D(\delta_t)$.
Therefore, for $\phi \in (-\alpha_{r \!, tv \!, \theta}, \alpha_{r \!, tv \!, \theta})$,

%assuming that
%$\tilde T_{S,\mathsf{f}^{(0)}(S)}$
%is uniformly distributed on $[t,t+\delta_t]$,
%we have
%
%
\begin{fleqn}
\begin{align*}
&\Prb\bigl( \Theta_{\mathsf{t}(I)}(\bm{u}_t) \in [\phi, \phi + \! \dd\phi] \,\,\big|\ \mathcal{N}_0, X_I = (r, \theta), T_{I, \mathsf{t}(I)} \! \in (t, t + \! \delta_t] \bigr)
\nonumber\\
&\approx
\frac
  {\bigl| \,\{ \bm{z} \in \tilde{D}(\delta_t): \theta_{\bm{z}} \in [\phi, \phi + \dd\phi] \}\, \bigr|}
  {\big| \, \tilde{D}(\delta_t) \, \big|}
  + o(1)\dd\phi
%  \quad \dd\phi \rightarrow 0,
\nonumber\\
&=
\frac
  {
   \left(
     \int_{\hat{r}(t)}^{\hat{r}(t) + h_\phi \!(\delta_t)}
       r_{\bm{z}}
     \dd r_{\bm{z}}
%     + o(1)
   \right)
   \dd\phi
  }
  {
   \int_{-\alpha_{r \!, tv \!, \theta}}^{\alpha_{r \!, tv \!, \theta}}
     \int_{\hat{r}(t)}^{\hat{r}(t) + h_{\theta_{\bm{z}}}\!(\delta_t)}
       r_{\bm{z}}
     \dd r_{\bm{z}}
   \dd\theta_{\bm{z}}
  }
  + o(1)\dd\phi
  +\! o(\dd\phi) 
%  \quad \dd\phi \rightarrow 0,
\nonumber\\
&=
\frac
  {
   \left(
     \cos\phi - \hat{a}\cos\varphi_{r \!, tv \!, \theta} + o(1)
   \right)
   \dd\phi
  }
  {
   2\sin\alpha_{r \!, tv \!, \theta} - 2\hat{a}\alpha_{r \!, tv \!, \theta}\cos\varphi_{r \!, tv \!, \theta} + o(1)
  }
  +  o(1)\dd\phi +\! o(\dd\phi) ,
\end{align*}
\end{fleqn}
%
%where we used that 
%$\Prb(\Phi(D(\delta_t))>1\mid X_I=(r,\theta),\,T_{I,\mathsf{t}(I)}\in(t,t+\delta_t])=o(1)$ as $\delta_t\to0$ and the approximation \eqref{eq:D_delta_ell_approx},
%a first-order approximation of the numerator in $\dd\phi$ in the first equality,
%and \eqref{eq:h_thetaz_sol} in the last equality.
where the first step follows from 
$\Prb(\Phi(D(\delta_t))>1\mid \mathcal{N}_0, X_I=(r,\theta),\,T_{I,\mathsf{t}(I)}\in(t,t+\delta_t])=o(1)$ as $\delta_t\to0$
and the approximation \eqref{eq:D_delta_ell_approx},
the second step uses a first-order approximation of the numerator in $\dd\phi$,
and the last step follows from \eqref{eq:h_thetaz_sol}.
Here $o(1)$ (resp., $o(\dd\phi)$) is with respect to 
$\delta_t \to 0$ (resp., $\dd\phi \to 0$).
Dividing both sides by $\dd\phi$ and taking the limits
$\dd\phi \to 0$ followed by $\delta_t \to 0$, we obtain \eqref{eq:dens_Theta_tS}.
\qed\\

\subsubsection{Proof of Lemma~\ref{lem:cond_expect_SINR}} \label{asec:lem_cond_expect_SINR}

By the definition in \eqref{eq:downlink_RLQ} along with the event $\mathcal{A}$ in \eqref{eq:set_conditions}, $\rho(s \mid \mathcal{A}) = \Exp[\mathsf{SINR}_I(t) \mid \mathcal{N}_0, \mathcal{A}]$.
Thus, using \eqref{eq:SINR},
%By the definitions of $\mathsf{RLQ}_{I}(s \mid \mathcal{N}_0)$ and $\mathcal{A}$ in \eqref{eq:downlink_RLQ} and \eqref{eq:set_conditions}, 
%under $\mathcal{A}$, we have
%We use the techniques in \cite{elsawy2013stochastic} to evaluate the conditional expectation of the SINR. Then, 
%$\Exp[\mathsf{SINR}_I(s) \mid \mathcal{A}] = \int_0^\infty \Prb(\mathsf{SINR}_I(s) > z' \mid \mathcal{A}) \dd z'$, 
%where $\mathsf{SINR}_I(s)$ is given in \eqref{eq:SINR}.
%Then, 
%
\begin{align}\label{eq:cond_SINR_eval}
&\rho(s \mid \mathcal{A})
%\nonumber\\
%&
%= \Exp[\mathsf{SINR}_I(s) \mid \mathcal{N}_0, \mathcal{A}]
\nonumber\\
&=
\int_0^\infty
\Prb\left(
  \frac
    {H_{I, n(s)} \|(r, \theta) - \bm{u}_s\|^{-\beta}}
    {\sigma^2 + \mathcal{I}_{\text{agg}}}
  > z'
   \ \bigg| \ 
     \mathcal{N}_0, \mathcal{A}
\right)
\dd z'
%\\
%&=
%\int_0^\infty
%  \exp
%  \left(
%    -\sigma^2 z' w_{r \!, sv \!, \theta}^\beta
%  \right)
%  \Big|_{X_I = (r, \theta)}
%\\
%&\qquad\qquad\times
%  \Exp
%  \left[
%    \exp
%%    \ee
%%    ^{
%    \left(
%      -\mathcal{I}_{\text{agg}} \,z' w_{r \!, sv \!, \theta}^\beta
%    \right)
%%    }
%    \ \Big| \
%      \mathcal{A}
%  \right]
%\dd z'
\nonumber\\
&=
\int_0^\infty
  w_{r \!, sv \!, \theta}^{-\beta}
%  \Big|_{X_I = (r, \theta)}
  \exp
  \left(
    -\sigma^2 z
  \right)
  \Exp
  \left[
    \ee
    ^{
%    \left(
      -z \mathcal{I}_{\text{agg}}
%    \right)
    }
    \ \Big| \
      \mathcal{N}_0, \mathcal{A}
  \right]
\dd z,
\end{align}
where $\mathcal{I}_{\text{agg}}$ is defined as
\begin{equation}\label{eq:interference_agg}
  \mathcal{I}_{\text{agg}}
  =
  \sum_{j \in \N\backslash\{I\}} H_{j, n(s)} \|X_j - \bm{u}_s\|^{-\beta}.
\end{equation}
Besides, in the second equality, we apply \cite[Eq.~(5)]{elsawy2013stochastic} with $\frac{\mu}{P_t A} = 1$, 
along with the change of variables
\begin{equation*}
  z = z' \|(r, \theta) - \bm{u}_s\|^\beta = z' w_{r \!, sv \!, \theta}^\beta,
\end{equation*}
where $w_{r \!, sv \!, \theta} = \sqrt{r^2 + s^2v^2 - 2rsv\cos\theta}$.
Under $\mathcal{A}$, 
$\| X_{\mathsf{t}(I)} - \bm{u}_{t} \| = \hat{r}(t)$, and the polar angle of $X_{\mathsf{t}(I)} - \bm{u}_{t}$ is $\phi$.
This means that
\begin{align*}
 \| X_{\mathsf{t}(I)} - \bm{u}_s \|
 &=
 \sqrt{\hat{r}(t)^2 + (s - t)^2 v^2 - 2\hat{r}(t) (s-t)v \cos\phi}
 \\
 &=
 w_{\hat{r}(t), (s-t)v, \phi}.
\end{align*}
Therefore, using \eqref{eq:interference_agg} and the Laplace transform of $H_{j, n(s)} \sim \exp(1)$,
\begin{align}\label{eq:LaplaceT_Iagg_eval}
&\Exp
\left[
  \ee
  ^{
%  \left(
    -z \mathcal{I}_{\text{agg}}
%  \right)
  }
  \ \Big| \
    \mathcal{N}_0, \mathcal{A}
\right]
=
\Exp
\left[
  \prod_{j \in \N\backslash\{I\}}
    \frac
      {1}
      {1 + z \|X_j - \bm{u}_s\|^{-\beta}}
  \ \Bigg| \
    \mathcal{N}_0, \mathcal{A}
\right]
\nonumber\\
&=
\frac
  {1}
  {1 + z \, w_{\hat{r}(t), (s-t)v, \phi}^{-\beta}}
\nonumber\\
&\quad\times
\Exp
\left[
  \prod_{j \in \N\backslash\{I, \mathsf{t}(I)\}}
    \frac
      {1}
      {1 + z \|X_j - \bm{u}_s\|^{-\beta}}
  \ \Bigg| \
    \mathcal{N}_0, \mathcal{A}
\right].
\end{align}
%
%Since $\Phi \backslash \{X_I, X_{\mathsf{t}(I)}\}$ ranges over 
%$\R^2 \backslash b_{\bm{o}}(\hat{a}r) \backslash b_{\bm{u}_t}(\hat{r}(t))$
%under $\mathcal{A}$, 
By the definitions of the events in \eqref{eq:initial_non-triggering_event} and \eqref{eq:set_conditions}, $\Phi \backslash \{X_I, X_{\mathsf{t}(I)}\}$ lies in 
$\R^2 \backslash b_{\bm{u}_0}(\hat{a}r) \backslash b_{\bm{u}_t}(\hat{r}(t))$, where $\bm{u}_t = (tv, 0)$.
Thus, the generating functional of a PPP (e.g., \cite[Example~9.4(c)]{daley2008point}) gives
\begin{align}\label{eq:apply_PGFL}
&\Exp
\left[
  \prod_{j \in \N\backslash\{I, \mathsf{t}(I)\}}
    \frac
      {1}
      {1 + z \|X_j - \bm{u}_s\|^{-\beta}}
  \ \Bigg| \
    \mathcal{N}_0, \mathcal{A}
\right]
\nonumber\\
&=
\exp
\left(
  -\lambda z 
  \int_{\R^2 \backslash b_{\bm{u}_0}(\hat{a}r) \backslash b_{\bm{u}_t}(\hat{r}(t))}
    \frac{1}{z + \|\bm{x} - \bm{u}_s\|^\beta}
  \dd \bm{x}
\right)
\nonumber\\
&=
\ee^{-\pi \lambda K_\beta z^{2/\beta}}
\exp
\left(
  \lambda z 
  \int_{b_{\bm{o}}(\hat{a}r) \cup b_{(tv, 0)}(\hat{r}(t))}
    \frac{1}{z + \|\bm{x} - \bm{u}_s\|^\beta}
  \dd \bm{x}
\right),
\end{align}
where $K_\beta$ is given in \eqref{eq:K_beta}. The last line follows since
\begin{equation*}
\exp\left(
-\lambda z \!
\int_{\R^2}
  \frac{1}{z + \|\bm{x} - \bm{u}_s\|^\beta}
\dd \bm{x}
\right)
=
\exp\left(
-\lambda z \!
\int_{\R^2}
  \frac{1}{z + \|\bm{x}\|^\beta}
\dd \bm{x}
\right),
\end{equation*}
to which we apply the last equality of \cite[Eq.~(26)]{tokuyama2024periodic}.

Since $t$ lies in the support of the PDF in \eqref{eq:pdf_TS0}, we have either the case
$|\hat{a}r - tv| \le \hat{r}(t) \le \hat{a}r + tv$,
or
the case $\hat{a}r + tv < \hat{r}(t)$,
in which $b_{\bm{o}}(\hat{a}r) \subset b_{(tv, 0)}(\hat{r}(t))$.
Therefore, similar techniques to the ones in \cite[Eq.~(32)]{tokuyama2024periodic} yield
\begin{align*}
&\int_{b_{(tv, 0)}(\hat{r}(t)) \backslash b_{\bm{o}}(\hat{a}r)}
  \frac{1}{z + \|\bm{x} - \bm{u}_s\|^\beta}
\dd \bm{x}
\\
&=
2\int_{\hat{a}r}^{tv + \hat{r}(t)} \!\!\!\!
  \int_0^{\cos^{-1} \left(-1 \vee \frac{x^2 + t^2v^2 - \hat{r}(t)^2}{2 x tv} \right)} \!
    \frac
      {x}
      {z + w_{x \!, sv \!, \psi}^{\,\beta}}
  \dd \psi
\dd x.
\end{align*}
By the above equation together with
\[
\int_{b_{\bm{o}}(\hat{a}r)}
  \frac{1}{z + \|\bm{x} - \bm{u}_s\|^\beta}
\dd \bm{x}
=
2\int_0^{\hat{a}r} \!\!\!
  \int_0^\pi \!
    \frac
      {x}
      {z + w_{x \!, sv \!, \psi}^{\,\beta}}
  \dd \psi
\dd x,
\]
we have
\begin{equation}\label{eq:mu_st_eval}
\exp
\left(
  \lambda z 
  \int_{b_{\bm{o}}(\hat{a}r) \cup b_{(tv, 0)}(\hat{r}(t))}
    \frac{1}{z + \|\bm{x} - \bm{u}_s\|^\beta}
  \dd \bm{x}
\right)
=
\mu(z, s, t),
\end{equation}
where $\mu(z, s, t)$ is given in \eqref{eq:func_mu_st}.
Plugging \eqref{eq:LaplaceT_Iagg_eval}, \eqref{eq:apply_PGFL}, and \eqref{eq:mu_st_eval} into \eqref{eq:cond_SINR_eval} leads to the result.
\qed\\

\subsubsection{Proof of Theorem~\ref{thm:ping-pong}}\label{asec:ping-pong}

By \eqref{eq:joint_pdf_XI} and \eqref{eq:cond_prob_N0},
and using the PDFs in \eqref{eq:pdf_TS0} and \eqref{eq:dens_Theta_tS},
$\Prb(\mathcal{P})$ defined in \eqref{eq:ev_ping_pong} is expressed as
\begin{align}\label{eq:prob_P_expansion}
\Prb(\mathcal{P})
&=
\int_0^\infty \!\!\!
\int_{-\pi}^{\pi} \!
\int_0^\infty \!\!\!
\int_{-\alpha_{r \!, tv \!, \theta}}^{\alpha_{r \!, tv \!, \theta}} 
\lambda r \ee^{-\pi\lambda\hat{a}^2 r^2}
\nonumber\\
&\quad\times
f_{T_{I,\mathsf{t}(I)} \mid \mathcal{N}_0, X_I}(t \mid r,\theta)\,
f_{\Theta_{\mathsf{t}(I)} \mid \mathcal{N}_0, X_I, T_{I,\mathsf{t}(I)}}(\phi \mid r,\theta,t)
\nonumber\\
&\quad\times
\Prb\bigl(
  \mathcal{S}_0
  \cap
  \{T_{\mathsf{t}(I),I} = t + \Delta_{\mathsf{T}} < F_I\}
  \mid
  \mathcal{N}_0, \mathcal{A}
\bigr)
\nonumber\\
&\quad\times \dd \phi\, \dd t\, \dd \theta\, \dd r,
\end{align}
where $\mathcal{A}$ is given in \eqref{eq:set_conditions}. 
%%$
%%\tilde{\mathcal{A}}
%%=
%%\{
%%  X_I = (r, \theta), T_{I,\mathsf{t}(I)} = t, \Theta_{\mathsf{t}(I)} = \phi
%%\}
%%$.
%
%%We see from \eqref{eq:polar_tI_at_T} and \eqref{eq:R_tS} that $\tilde{\mathcal{A}} \subset \mathcal{A}$, which is defied in \eqref{eq:set_conditions}, and therefore $X_I$ and $X_{\mathsf{t}(I)}$ are given under $\tilde{\mathcal{A}}$.
%
Since $T_{\mathsf{t}(I), I}$ defined in \eqref{eq:TTT_expire_tI_to_I} is conditionally independent of $F_I$ and $\mathcal{S}_0$ under $\mathcal{A}$, we have
\begin{align}\label{eq:cond_indep_decomp}
&\Prb\bigl(
  \mathcal{S}_0
  \cap
  \{T_{\mathsf{t}(I),I} = t + \Delta_{\mathsf{T}} < F_I\}
  \mid
  \mathcal{N}_0, \mathcal{A}
\bigr)
\nonumber\\
&=
\Prb\bigl(T_{\mathsf{t}(I),I} = t + \Delta_{\mathsf{T}}\mid \mathcal{A}\bigr)
%\,
\Prb\bigl(
  \mathcal{S}_0
  \cap
  \{t + \Delta_{\mathsf{T}} < F_I\}
  \mid
  \mathcal{N}_0, \mathcal{A}
\bigr),
\end{align}
where the conditioning on $\mathcal{N}_0$ is omitted in the first term, since $T_{\mathsf{t}(I),I}$ depends only on $X_I$, $X_{\mathsf{t}(I)}$, and $T_{I,\mathsf{t}(I)}$.

By \eqref{eq:TTT_expire_tI_to_I}, under $\mathcal{A}$,
$T_{\mathsf{t}(I),I} = t + \Delta_{\mathsf{T}}$
holds if and only if
\begin{equation}\label{eq:ping-pong_equiv1}
\hat{a} \|X_{\mathsf{t}(I)} - \bm{u}_s\| \geq \|X_I - \bm{u}_s\|
\quad \forall s \in [t + \Delta_{\mathsf{T}}, t + 2\Delta_{\mathsf{T}}).
\end{equation}
Let $\tilde{\bm{x}}_{I} = X_I - \bm{u}_t$, and $\tilde{\bm{x}}_{\mathsf{t}(I)} = X_{\mathsf{t}(I)} - \bm{u}_t$.
%\[
%  \bm{x}_{I} := X_I - \bm{u}_t,
%  \qquad
%  \bm{x}_{\mathsf{t}(I)} := X_{\mathsf{t}(I)} - \bm{u}_t.
%\]
%%
Under $\mathcal{A}$, they are expressed as
\begin{align*}
  &\tilde{\bm{x}}_I = (r, \theta) - (tv, 0) = (w_{r \!, tv \!, \theta}, \varphi_{r \!, tv \!, \theta}),
  %\label{eq:x_I_fixed}
  \\
  &\tilde{\bm{x}}_{\mathsf{t}(I)} = (\hat{a} w_{r \!, tv \!, \theta}, \phi),
  %\label{eq:x_tI_fixed}
\end{align*}
where \eqref{eq:r_hat_t} is used in the second line, and $\varphi_{r \!, tv \!, \theta}$ is given in \eqref{eq:vphi_cosinv}.
%, with 
%$w_{r \!, tv \!, \theta} = \sqrt{r^2 + t^2 v^2 - 2rtv\cos\theta}$.
%
%
Substituting these expressions into \eqref{eq:ping-pong_equiv1} and using $\bm{u}_s = (sv, 0)$ yields
\begin{equation*}%\label{eq:ping-pong_equiv2}
\hat{a} \|\tilde{\bm{x}}_{\mathsf{t}(I)} - (sv, 0)\| \geq \|\tilde{\bm{x}}_I - (sv, 0)\|
\quad \forall s \in [\Delta_{\mathsf{T}}, 2\Delta_{\mathsf{T}}).
\end{equation*}
This implies that
%%, under $\tilde{\mathcal{A}}$,
%the event
%\[
$\{ T_{\mathsf{t}(I),I} = t + \Delta_{\mathsf{T}} \mid \mathcal{A}\}$
%\]
is equivalent to the event that the interval
$[\Delta_{\mathsf{T}} v, 2\Delta_{\mathsf{T}} v)$
does not intersect the following region
%$\tilde{D}_{\mathsf{t}(I), I}$ defined as 
%
\begin{equation*}
  \tilde{D}_{\mathsf{t}(I), I} 
  = 
  \left\{
    \bm{z} \in \R^2: 
    \hat{a} \|\tilde{\bm{x}}_{\mathsf{t}(I)} - \bm{z}\| 
    < 
    \|\tilde{\bm{x}}_{I} - \bm{z}\| 
  \right\},
\end{equation*}
which is the interior of a circle with radius $d_{\! \mathcal{A}}$
given in \eqref{eq:radius_apollonius} and centered at
$(x_{\! \mathcal{A}}, y_{\! \mathcal{A}})$ in Cartesian coordinates,
as defined in \eqref{eq:center_apollonius}\footnote{These results of the radius and the center can be derived using elementary geometry; the derivation is omitted for conciseness.}.
%
%
%
%
%
%Under $X_I = \tilde{\bm{x}}_I + \bm{u}_t$ and $X_{\mathsf{t}(I)} = \tilde{\bm{x}}_{\mathsf{t}(I)} + \bm{u}_t$,
%$T_{\mathsf{t}(I), I}$ is conditionally independent of the events $F_I$ and $\mathcal{S}_0$ defined in \eqref{eq:RLF_expire_i} and \eqref{eq:initial_in-sync_event}.
%
%
%
%
%\textcolor{blue}
%{
%Thus,
%%
%\begin{align}\label{eq:cond_indep_decomp}
%&\Prb\bigl(
%  \mathcal{S}_0
%  \cap
%  \{T_{\mathsf{t}(I),I} = t + \Delta_{\mathsf{T}} < F_I\}
%  \mid
%  \mathcal{A}
%\bigr)
%\nonumber\\
%&=
%\Prb\bigl(T_{\mathsf{t}(I),I} = t + \Delta_{\mathsf{T}}\mid \mathcal{A}\bigr)
%\,
%\Prb\bigl(
%  \mathcal{S}_0
%  \cap
%  \{t + \Delta_{\mathsf{T}} < F_I\}
%  \mid
%  \mathcal{A}
%\bigr),
%\end{align}
%%
%}
%%
Therefore, by
$q_{\! \mathcal{A}}^\pm = x_{\! \mathcal{A}} \pm \sqrt{d_{\! \mathcal{A}}^2 - y_{\! \mathcal{A}}^2}$,
\begin{align}\label{eq:T_tI_cond_prob_closed}
  &\Prb\bigl(T_{\mathsf{t}(I), I} = t + \Delta_{\mathsf{T}} \,|\, \mathcal{A} \bigr)
%  \nonumber\\
%  &=
%  \Prb(d_{\! \mathcal{A}} < y_{\! \mathcal{A}})
%  +
%  \Prb\bigl(
%    d_{\! \mathcal{A}} > y_{\! \mathcal{A}}, 
%    [\Delta_{\mathsf{T}} v, 2\Delta_{\mathsf{T}} v)
%    \cap
%    (q_{\! \mathcal{A}}^-, q_{\! \mathcal{A}}^+)
%    =
%    \emptyset
%  \bigr)
  \nonumber\\
  &=
  1 -
  H\bigl(
    \min\{ 2\Delta_{\mathsf{T}} v,\, q_{\! \mathcal{A}}^+ \}
    -
    \max\{ \Delta_{\mathsf{T}} v,\, q_{\! \mathcal{A}}^- \}
  \bigr)
  H\bigl(
    d_{\! \mathcal{A}} - y_{\! \mathcal{A}}
  \bigr),
\end{align}
which follows from the fact that
$[a, b] \cap [c, d] \neq \emptyset$ if and only if
$\max\{a, c\} \le \min\{b, d\}$.

%Next consider $\mathcal{S}_0 \cap \{t + \Delta_{\mathsf{T}} < F_I\}$ under $\mathcal{N}_0$ and $\mathcal{A}$.
%
%By Assumption~\ref{assump:cond_expect_SINR} and the definitions of $\mathcal{S}_0$ and $F_I$ in \eqref{eq:initial_in-sync_event} and \eqref{eq:RLF_expire_i},
By the definitions in \eqref{eq:RLF_expire_i} and \eqref{eq:initial_in-sync_event},
under $\mathcal{N}_0$ and $\mathcal{A}$, $\mathcal{S}_0 \cap \{t + \Delta_{\mathsf{T}} < F_I\}$ occurs if and only if
\begin{align*}
  &\rho(0 \mid \mathcal{A}) \ge Q_{\mathsf{F}},
  \\
  &\text{and either }
  t + \Delta_{\mathsf{T}} < \Delta_{\mathsf{F}},
  \\
  &\text{or }
  \Bigl\{
    t + \Delta_{\mathsf{T}} \ge \Delta_{\mathsf{F}}
    \ \text{ and } 
    \\
    &\quad
    \forall \tau \in [\Delta_{\mathsf{F}}, t + \Delta_{\mathsf{T}}],\
    \exists s \in [\tau-\Delta_{\mathsf{F}}, \tau):
    \rho(s \mid \mathcal{A}) \ge Q_{\mathsf{F}}
  \Bigr\}.
\end{align*}
%
%
%For analytical tractability, we approximate the above interval-wise condition by evaluating the radio link quality at the beginning of the interval and at $(t+\Delta_{\mathsf{T}})$, i.e.,
%\[
%\rho(t+\Delta_{\mathsf{T}}-\Delta_{\mathsf{F}}\mid\mathcal{A})\ge Q_{\mathsf{F}}
%\quad\text{and}\quad
%\rho(t+\Delta_{\mathsf{T}}\mid\mathcal{A})\ge Q_{\mathsf{F}},
%\]
By Approximation~\ref{approx:cond_expect_SINR}, under $\rho(0 \mid \mathcal{A}) \ge Q_{\mathsf{F}}$, it holds that
\[
\begin{aligned}
&\forall \tau\in[\Delta_{\mathsf{F}},t+\Delta_{\mathsf{T}}],
\exists s\in[\tau-\Delta_{\mathsf{F}},\tau):
\rho(s\mid\mathcal{A})\ge Q_{\mathsf{F}},
\\
&\quad\approx
%\Bigl\{
\rho(t+\Delta_{\mathsf{T}}-\Delta_{\mathsf{F}}\mid\mathcal{A})
\ge Q_{\mathsf{F}}.
%\Bigr\}
%\ \text{or} \
%\Bigl\{
%\rho(t+\Delta_{\mathsf{T}}\mid\mathcal{A})
%\ge Q_{\mathsf{F}}
%\Bigr\}.
\end{aligned}
\]
Therefore, we have
\begin{align}\label{eq:S0_FI_cond_prob_closed}
&\Prb\bigl(
  \mathcal{S}_0
  \cap
  \{t + \Delta_{\mathsf{T}} < F_I\}
  \mid
  \mathcal{N}_0, \mathcal{A}
\bigr)
\nonumber\\
&=
%H\bigl(
%\min\{
%  \zeta_{\mathcal{A}}(0),\;
%  \zeta_{\mathcal{A}}(t + \Delta_{\mathsf{T}} - \Delta_{\mathsf{F}})
%\}
%-
%Q_{\mathsf{F}}
%\bigr)
H\bigl(
\zeta_{\mathcal{A}}(0)
-
Q_{\mathsf{F}}
\bigr)
H\bigl(
J_{\mathcal{A}}(t)
-
Q_{\mathsf{F}}
\bigr),
\end{align}
where $\zeta_{\mathcal{A}}(0)$ and $J_{\mathcal{A}}(t)$ are given by \eqref{eq:func_zeta} and approximately by \eqref{eq:func_J_A}, respectively.

Consequently, substituting \eqref{eq:cond_indep_decomp}, 
\eqref{eq:T_tI_cond_prob_closed}, and \eqref{eq:S0_FI_cond_prob_closed} into 
\eqref{eq:prob_P_expansion}, and combining with \eqref{eq:dens_Theta_tS}, 
the result follows by changing the variable to $\theta' = |\theta|$, 
since the integrand is even in $\theta$.
\qed\\

\subsubsection{Proof of Corollary~\ref{cor:too-late}}\label{asec:too-late}

By the definition in \eqref{eq:ev_too_late}, 
$\Prb(\mathcal{L})$ can be expressed similarly to \eqref{eq:prob_P_expansion} as
\begin{align}\label{eq:prob_L_expansion}
\Prb(\mathcal{L})
&=
\int_0^\infty \!\!\!
\int_{-\pi}^{\pi} \!
\int_0^\infty \!\!\!
\int_{-\alpha_{r \!, tv \!, \theta}}^{\alpha_{r \!, tv \!, \theta}} 
\lambda r \ee^{-\pi\lambda\hat{a}^2 r^2}
\nonumber\\
&\quad\times
f_{T_{I,\mathsf{t}(I)} \mid \mathcal{N}_0, X_I}(t \mid r,\theta)\,
f_{\Theta_{\mathsf{t}(I)} \mid \mathcal{N}_0, X_I, T_{I,\mathsf{t}(I)}}(\phi \mid r,\theta,t)
\nonumber\\
&\quad\times
\Prb\bigl(
  \mathcal{S}_0
  \cap
  \{t + \Delta_{\mathsf{T}} \ge F_I\}
  \mid
  \mathcal{N}_0, \mathcal{A}
\bigr)
\,\dd \phi\,\dd t\,\dd \theta\,\dd r.
\end{align}
%
%
%%
%$\Prb\bigl(
%  \mathcal{S}_0
%  \cap
%  \{T_{\mathsf{t}(I),I} = t + \Delta_{\mathsf{T}} < F_I\}
%  \mid
%  \mathcal{A}
%\bigr)$
%in \eqref{eq:prob_P_expansion} with 
%%
%\[
%\Prb\bigl(
%  \mathcal{S}_0
%  \cap
%  \{t + \Delta_{\mathsf{T}} \ge F_I\}
%  \mid
%  \mathcal{A}
%\bigr),
%\]
%%
%
%
%By letting $\theta' = |\theta|$, 
%the result follows by an argument analogous to \eqref{eq:S0_FI_cond_prob_closed}, yielding
Analogously to \eqref{eq:S0_FI_cond_prob_closed}, we have
\begin{align*}%\label{eq:S0_FI_cond_prob_closed_for_L}
&\Prb\bigl(
  \mathcal{S}_0
  \cap
  \{t + \Delta_{\mathsf{T}} \ge F_I\}
  \mid
  \mathcal{N}_0, \mathcal{A}
\bigr)
\nonumber\\
&=
H\bigl(
\zeta_{\mathcal{A}}(0)
-
Q_{\mathsf{F}}
\bigr)
\,
\left[
1 - 
H\bigl(
%\zeta_{\mathcal{A}}(t + \Delta_{\mathsf{T}} - \Delta_{\mathsf{F}})
J_{\mathcal{A}}(t)
-
Q_{\mathsf{F}}
\bigr)
\right]
\nonumber\\
&=
H\bigl(
\zeta_{\mathcal{A}}(0)
-
Q_{\mathsf{F}}
\bigr)
\,
\bar{H}\bigl(
%\zeta_{\mathcal{A}}(t + \Delta_{\mathsf{T}} - \Delta_{\mathsf{F}})
J_{\mathcal{A}}(t)
-
Q_{\mathsf{F}}
\bigr).
\qedhere
\end{align*}
Plugging the above equation into \eqref{eq:prob_L_expansion}, together with \eqref{eq:dens_Theta_tS}, and applying the change of variable $\theta' = |\theta|$, we obtain the result.
\qed

%%% コメントアウト（Appendix~B）―終了― %%%
%\end{comment}

%\section*{Acknowledgment}
%This work was supported in part by a research grant from the
%Okasan-Kato Foundation, Tsu, Japan.
%The author would like to thank Prof. Naoto Miyoshi for valuable discussions.

\end{document}